\documentclass[reqno,12pt,letterpaper,oneside]{amsart}
\usepackage[utf8]{inputenc}
\usepackage{style}
\let\epsilon=\varepsilon

\usepackage{float}
\usepackage[font=small,labelfont=bf]{caption}
\usepackage{mathtools}
\usepackage{algorithm}
\usepackage{algpseudocode}
\newcommand{\mat}[1]{\begin{pmatrix} #1 \end{pmatrix}}
\usepackage{float}
\usepackage{bm}
\newcommand{\set}[1]{ \left \{ #1 \right \}}
\newtheorem*{thm*}{Theorem}
\newtheorem{defi}[prop]{Definition}

\numberwithin{equation}{section}
\newcommand*{\ddd}{\mathop{}\!\mathrm{d}}

\usepackage[style=alphabetic, backend=biber,doi=false,isbn=false,url=false,eprint=false,giveninits=true]{biblatex}
\title{Particle Trajectories for Quantum Maps}
\author{Yonah Borns-Weil}
\email{yonah\_borns-weil@berkeley.edu}
\address{Department of Mathematics, University of California, Berkeley,
CA 94720}
\author{Izak Oltman}
\email{ioltman@berkeley.edu}
\address{Department of Mathematics, University of California, Berkeley,
CA 94720}
\keywords{}
\subjclass[]{}
\begin{document}
\begin{abstract}
We study the trajectories of a semiclassical quantum particle under repeated indirect measurement by Kraus operators, in the setting of the quantized torus. In between measurements, the system evolves via either Hamiltonian propagators or metaplectic operators. We show in both cases the convergence in total variation of the quantum trajectory to its corresponding classical trajectory, as defined by the propagation of a semiclassical defect measure. This convergence holds up to the Ehrenfest time of the classical system, which is larger when the system is ``less chaotic". In addition, we present numerical simulations of these effects.

In proving this result, we provide a characterization of a type of semi-classical defect measure we call uniform defect measures. We also prove derivative estimates of a function composed with a flow on the torus.
\end{abstract}
\maketitle


\section{Introduction}\label{s:intro}

\begin{figure}
    \centering
    \includegraphics[width = \textwidth]{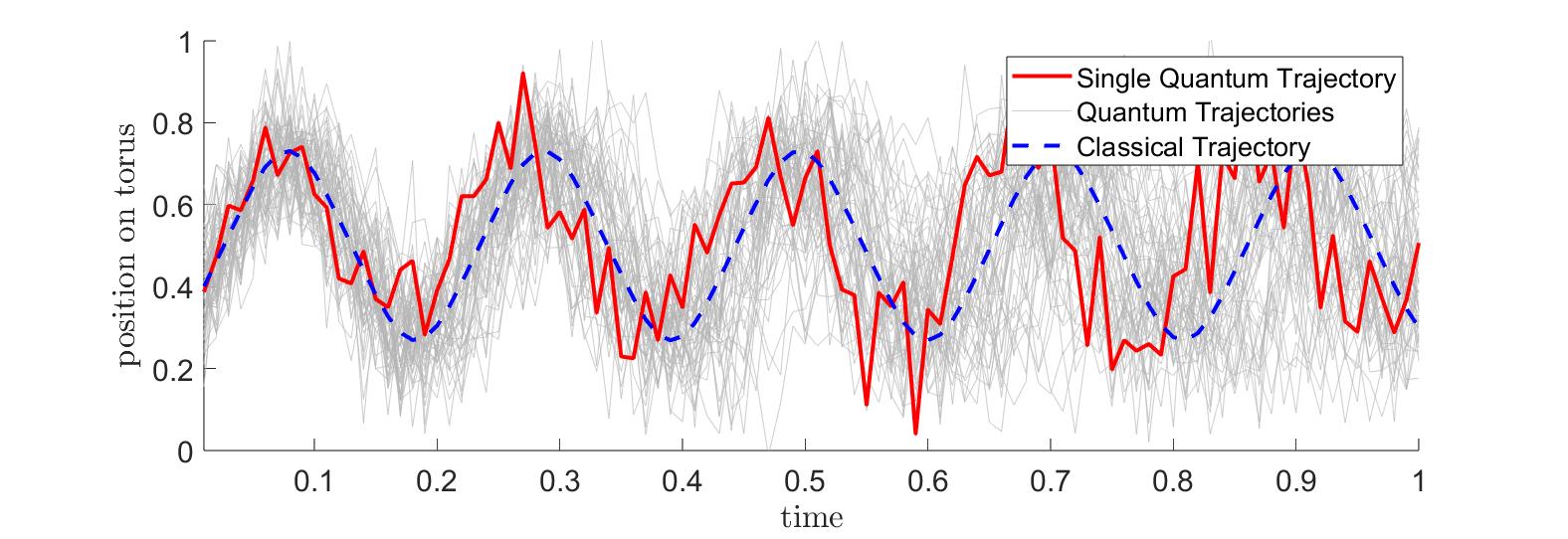}
    \caption{A numerical simulation demonstrating this paper's main result. Each grey curve represents the trajectory of a quantum particle (under the evolution procedure discussed in this paper) on a torus with identical initial conditions and a fixed Hamiltonian. A single trajectory is plotted in red and the corresponding classical trajectory is plotted as a blue dotted line. Up until time $0.4$, most of the positions of the quantum particles are within $0.1$ of the classical trajectory, but shortly after become decoherent. More numerics are provided in \S \ref{s:numerics}. This simulation is discussed at the end of \S \ref{ss:model} and \S \ref{s:numerics}.}
    \label{fig:fig_first_page}
\end{figure}

The framework of ``quantum trajectories" has gotten increasing attention in recent years. Instead of treating measurement as a one-time event, researchers consider many measurements obtained over time and study their effects on the particle being measured, described mathematically via quantum instruments. The sequence of measured values is called the \emph{quantum trajectory}, and is a random variable worthy of study.

In this paper, we take quantum trajectories into the setting of the semiclassical quantized torus, as studied in Bouzouina–De Bi\'evre \cite{BD}, Schenck \cite{Schenck}, Dyatlov–Jezequel \cite{DJ}, and others. The spaces involved are all finite-dimensional, which makes numerical simulations easier than in PDE-based models. Our result is inspired by the recent work of Benoist, Fraas, and Fr\"ohlich \cite{BFF}, who prove in the PDE setting that the trajectory of a repeatedly observed quantum particle in the semiclassical limit approaches the natural notion of a classical trajectory. See Figure \ref{fig:fig_first_page} for a simulation of this quantum-classical correspondence of particle trajectories on the quantized torus and \S \ref{s:numerics} for more numerics. We strengthen the result in \cite{BFF} in our setting by allowing the time to depend on the semiclassical parameter, and we relate the rate of convergence to the ``chaotic behavior" of the underlying classical system. This matches the intuition that a less chaotic transformation should be well-approximated for longer by its classical counterpart. Additionally, we illustrate the results numerically to demonstrate in the chaotic case that our time restriction is essentially optimal.

\subsection{Model}\label{ss:model}

In our setting, a quantum particle is modeled by a density operator (see Definition \ref{def:desnity operator}) $ \rho_N$ on the quantized $2d-$dimensional torus $H_N$ (see \eqref{eq:quantizedTorus}) where $N \in \mathbb {N}$ is proportional to the reciprocal of the semi-classical parameter $h$. This particle is ``indirectly measured'' by conjugating by Kraus operators (see Definition \ref{def:Kraus operators}) $\Op_N (f_q)$ where $\{f _q \} _{q \in \Omega}$ are fixed elements of $C^\infty (\mathbb {T}^{2d})$ with uniformly bounded derivatives (such that $\Op_N(f_q)$ provide a resolution of the identity), $\Omega$ is a compact metric space with finite Borel measure $\nu$, and $\Op_N$ is the quantization of functions on the torus as described in \S \ref{ss:quantizedtorus}. The measurement procedure is further described in \S \ref{s:measurement procedure}.

After observing the particle, we evolve it for one unit of time by conjugating $\rho_N$ by a unitary operator $U$, which is either a quantization of a fixed symplectic matrix $M$ (defined in \S \ref{s:Fourier integral operators}) or Schr\"odinger evolution $e^{-\frac{i}{h} P}$ for a fixed Hamiltonian $P$. If we repeat this process $n$ times, we get a resulting measure on the space of trajectories. Suppose $F \subset \Omega ^{n+1}$, then the probability of obtaining the set of trajectories $F$ is given by:
\begin{align*}
    \int_F\tr_{H_N}\left(\Phi_{N,{q_n}}^*\circ\dots\circ\Phi_{N,{q_0}}^*(\rho_N)\right)\ddd \nu(q_0)\cdots\ddd \nu(q_n)
\end{align*}
where
\begin{align*}
    \Phi_{N,{q_i}}^*(\rho):=U\Op_N(f_{q_i})\rho_N\Op_N(f_{q_i})^*U^*.
\end{align*}
Therefore, we have a ``quantum probability measure" given by \begin{equation}\label{eq:quantummeasure}
\bP^{(n)}_{N,\rho}:=\tr_{H_N}\left(\Phi_{N,{q_n}}^*\circ\dots\circ\Phi_{N,{q_0}}^*\rho\right)\ddd \nu(q_0)\cdots\ddd \nu(q_n)
\end{equation}

We now describe the trajectory of a corresponding classical particle. Let $\mu$ be the defect measure (defined in Definition \ref{def:defect measure}) of $\rho_N$. We interpret $\mu$ as the probability distribution of the initial position and momentum of a classical particle. We apply an ``approximate measurement," which computes an observable quantity $q$ to be in $E_0\subseteq \Omega$ with probability $\int_{E_0}|f_{q_0}(\zeta)|^2\ddd \nu(q_0).$ The particle is then allowed to classically evolve for one unit of time by a flow $\phi^t $ on $\mathbb T^{2d}$ which is either multiplication by the symplectic matrix $M$ or else the nonlinear Hamiltonian evolution $\text{exp}(H_p)$. We again repeat the process of alternating measurement and evolution, though we must remark that unlike in the quantum case, the measurement does not affect the location of the particle. Conditioned on initially having position/momentum $\zeta$, we have the probability of classically measuring $(q_0,\dots q_n)$ in a measurable set $F\subset \Omega^{n+1}$ to be $$\int_F|f_{q_n}(\phi^n(\zeta))|^2\cdots|f_{q_0}(\zeta)|^2\ddd \nu(q_0)\cdots\ddd \nu(q_n).$$ One can easily check that $(q_n)_n$ is a process of independent random variables (although not identically distributed). As the initial value of $\zeta$ was taken to be randomly chosen from the distribution $\mu$, we take the classical probability of measuring $(q_0,\dots ,q_n)\in F $ to be $$\int_{\T^{2d}}\int_F|f_{q_n}(\phi^n(\zeta))|^2\cdots|f_{q_0}(\zeta)|^2\ddd \nu(q_0)\cdots\ddd \nu(q_n)\ddd \mu(\zeta).$$

We therefore have a ``classical probability measure" given by \begin{equation}\label{eq:classicalmeasure}
P^{(n)}_{\mu}:=\int_{\T^{2d}}|f_{q_n}(\phi^n(\zeta))|^2\cdots|f_{q_0}(\zeta)|^2\ddd \mu(\zeta)\ddd \nu(q_0)\cdots\ddd \nu(q_n).\end{equation}
Note that the functions $f_q$ are allowed to overlap so that $P^{(n)}_{\mu}$ is a probability of trajectories in $\Omega^{n+1}$ while $\phi^t$ is the unique trajectory on the torus.

The goal of this paper is to describe in what sense, and at what quantitative rate, the measure $\bP^{(n)}_{N,\rho}$ approaches $P^{(n)}$ as $N$ grows large.

We can state our main result.

\begin{thm*}[Main result]\label{t:abridged}
Suppose $\rho=\rho_N:H_N\to H_N$ is a density operator with semiclassical defect measure $\mu$. Fix $n\in \N$, and let $(q_{N,0},q_{N,1}, \dots , q_{N,n}) $ be random variables with law $\bP_{N,\rho}^{(n)}$ given by \eqref{eq:quantummeasure} and let $(q_0, q_1,\dots , q_n)$ be random variables with law $P_{\mu}^{(n)}$ given by \eqref{eq:classicalmeasure}. Then as $N\to \infty$, $(q_{N,0},q_{N,1}, \dots , q_{N,n}) $ converges weakly to $(q_0, q_1,\dots , q_n)$.
\end{thm*}

In proving this, we show uniform convergence of $\bP_{N,\rho}^{(n)}$ to $P_{\mu}^{(n)}$ and hence $$\bP_{N,\rho}^{(n)}\xrightarrow{N\to\infty}P_{\mu}^{(n)}$$ in total variation. For a stronger version of Theorem \ref{t:abridged}, see Theorem \ref{t:torusquant}, in which we provide a quantitative rate of convergence in the case where the number of steps $n$ approaches the Ehrenfest time (defined in \ref{eq:Ehrenfest_time}). 

Figure \ref{fig:fig_first_page} numerically samples 60 quantum trajectories from such an evolution procedure (see the end of \S \ref{s:numerics} for more details) for $100$ time steps of $0.01$ units of time. Observe that up until $t =0.4$, most trajectories are within $0.1$ of the classical trajectory\footnote{Our result shows that before the Ehrenfest time, at each fixed time, the distribution of quantum positions is a Gaussian centered at the classical trajectory position with variance given by the precision of the indirect measurement.}, while shortly after this time, they disagree. This paper's main result estimates how long they agree, and to what distribution they agree with.
\subsection{Previous Work}

The understanding that a measurement affects a quantum state goes back to the early days of quantum mechanics. In his paper on the uncertainty principle, Heisenberg \cite{Heisenberg} (see \cite{Heisenberg_English} for an English translation) discussed the so-called ``collapse" of the wave function, which was later rephrased in mathematical terms by von Neumann \cite{VN}. When a system is only partially or indirectly measured, the necessary framework is that of positive operator-valued measures, which were first studied by Naimark \cite{Naimark}. The analog of wave-function collapse was introduced by Davies and Lewis \cite{DL} with the development of quantum instruments.

Further research focused on multiple measurements as a probabilistic process, which could be modeled in either discrete or continuous time. The first analysis of continuous-time measurements was due to Davies \cite{Davies}, and required extensive machinery from stochastic calculus (See \cite{BG} or \cite{Holevo} for a brief introduction). Meanwhile, the discrete-time model consisting of repeated applications of a quantum instrument, while requiring substantially fewer technicalities, has nonetheless demonstrated a rich variety of behavior and is a subject of current study. K\"ummerer and Maassen \cite{KM_ergodic} demonstrated ergodicity of a repeated measurement process and showed that the states approach a pure state provided the measurement operators do not have a ``dark" subspace \cite{KM_purification}. Ballesteros et al.\ \cite{BCFFS} show that when the dynamic is trivial, the state of the system localizes in space.

If the state is allowed to evolve in-between measurements, the situation becomes more complex. The mathematical analysis of such a quantum trajectory, which approximates a classical one, was done by Ballesteros et al.\ \cite{BBFF} following extensive physical evidence of the phenomenon, see the exposition by Figari \cite{Figari}. That paper studied particles, initially in the set of normal states, evolving under a quadratic Hamiltonian. The motivation for our work came from the study of the semiclassical case under a more general Hamiltonian by Benoist, Fraas, and Fr\"ohlich in \cite{BFF}. We adapt their framework to the setting of quantum maps.
This has two advantages: we can improve on the number of measurements exhibiting classical/quantum correspondence and we can provide (more easily than in the PDE setting) numerical simulations. In particular, the numerics indicate the accuracy of the dynamical bound on the number of measurements.

\section{Background}\label{s:background}

In this section, we provide a background on the measurement procedure as well as the semiclassical analysis tools required to prove our main result. Throughout this paper, we consider $h  = (2\pi N)^{-1} >0$ as the semiclassical parameter where $N \in \mathbb N$.

\subsection{Semiclassical Analysis on the Real Line}\label{ss:semiclassical}

Physically, we are often concerned with how our quantum mechanical world approximately produces classical mechanics. This limit is achieved by taking a semiclassical parameter $h\in(0,1)$ to be very small, and the math involved is called \emph{semiclassical analysis}.

\begin{defi}[symbol class]\label{def:symbol class}
For each $\delta \in [0,1/2)$, the \emph{symbol class} $S_\delta$ is defined as $$S_{\delta}:=\{a\in\cinf(T^*\R^d):\forall \alpha,\beta\in\N^d,\, |\dd_x^{\alpha}\dd_{\xi}^{\beta}a(x,\xi)|\le C_{\alpha,\beta}h^{-\gd|\ga+\gb|}\}.$$
\end{defi}

 Of most importance is the case where $\gd=0$, though technical steps of the proof of Theorem \ref{t:torusquant} will require us to consider general $\gd$. Nevertheless, if $\gd$ is omitted it will be taken to be $0$. Elements of $S_{\gd}$ (called \emph{symbols}) may depend on $h$, though we require that the constants $C_{\alpha,\beta}$ in the definition be uniform in $h$. Real-valued symbols correspond to classical observables, which can be ``quantized" to get quantum observables as follows.

\begin{defi}[Weyl quantization]
For $\delta \in [0,1/2)$, the \emph{Weyl quantization} of a symbol $a\in S_\delta$ is an operator $\Op_ha:\Sc(\R^d)\to \Sc'(\R^d)$ (sometimes written $a(x,hD)$) given by
\begin{equation*}(\Op_h(a)u)(x):=\frac{1}{(2\pi h)^d}\iint e^{\frac{i}{h}\la x-y,\xi\ra}a\left(\frac{x+y}{2},\xi\right)u(y) \ddd y \ddd \xi,\end{equation*}
for $u \in \Sc(\R^d)$.
\end{defi}

A theorem of Calder\'on and Vaillancourt states that if $\delta \in [0,1/2)$ and $a\in S_{\gd}$, then $\Op_h(a)$ is a bounded operator on $L^2(\R^d)$, with bound \begin{equation}\label{eq:CV}\|\Op_h(a)\|_{L^2(\R^d)\to L^2(\R^d)}\le C\|a\|_{L^{\infty}(\R^d)}+C\sum_{|\alpha| + |\beta|\le K}h^{\frac{|\ga|+|\beta|}{2}}\| \partial_x^\alpha \partial_\xi^\beta a\|_{L^{\infty}(\R^d)}\end{equation} where $C>0$ and $K\in\N$ depends only on $d$. (In fact, we remark that (\ref{eq:CV}) even holds for the threshold scaling $\delta=\frac{1}{2}$, though in this case, all derivatives of $a$ in the sum will contribute equal orders of magnitude.)

If $a\in S$ is real-valued, $\Op_h(a)$ is a self-adjoint operator on $L^2(\R^d)$, and more generally we have $$\Op_h(a)^*=\Op_h(\overline{a}).$$

Weyl quantizations enjoy good properties under composition. In particular, if $\gd\in[0,\frac{1}{2})$ and $a,b \in S_{\gd}$, then $$\Op_h(a)\Op_h(b)=\Op_h(a\#b)$$ where $a\#b\in S_{\gd}$ is the \emph{Moyal product} given explicitly by \begin{equation*}
 a\#b(x,\xi):= e^{ih\sigma(D_x,D_{\xi},D_y,D_{\eta})}(a(x,\xi)b(y,\eta))\Big|_{\substack{y=x\\\eta=\xi}}
\end{equation*}
from which one may derive the asymptotic summation formula
\begin{equation}\label{eq:comp_rule}
 a\#b(x,\xi)\sim\sum_{k=0}^{\infty}\frac{i^kh^k}{k!}\sigma(D_x,D_{\xi},D_y,D_{\eta})^k(a(x,\xi)b(y,\eta))\Big|_{\substack{y=x\\\eta=\xi}}
\end{equation}
where in these formulas $\gs(x,\xi,y,\eta):=\la\xi,y\ra-\la x,\eta\ra$, and $a\sim\sum_{k=0}^{\infty}h^ka_k$ means that for each $K \in \N$, $(a-\sum_{k=0}^{K-1}h^ka_k )\in h^K S_\delta .$

In particular, one has that if $a,b\in S_{\delta}$, then $(\Op_ha)(\Op_hb)-\Op_h(ab)=\Op_h(r)$ where $r$ has $S_{\delta}$ seminorms at most $O(h^{1-2\gd})$, with constants depending only on finitely many derivatives of $a$ and $b$. More generally, given $a_1,\dots, a_n\in S_{\gd}$, we have by an easy induction together with $(\ref{eq:CV})$ that
\begin{equation}\label{eq:n_comp_rule}
 \left\|\Op_h(a_1)\cdots\Op_h(a_n)-\Op_h(a_1\cdots a_n)\right\|_{L^2(\R^d)\to L^2(\R^d)}=\poly(n)O(h^{1-2\gd}).
\end{equation}

Finally, we will need a theorem of Egorov, which gives a sense in which quantum operators evolve in time analogously to their classical counterparts. We provide the quantitative version due to Bouzouina and Robert \cite{BR}, see also Zworski \cite[Chapter 11.4]{Zworski}. We first require the following definition.

\begin{defi}[Lyapunov exponent]\label{def:Lyapunov exponent}
    For any dynamical system given by a flow $\phi_t$ on $T^*\mathbb R^d$ or $\mathbb T^{2d}$ (smooth in time), define the \emph{Lyapunov exponent} $\Gamma$ to be $$\Gamma=\lim_{t\to\infty}\sup_y\frac{1}{t}\log|\dd_y\phi_t(y)|$$ and analogously for iterative maps. 
\end{defi}

Given $p\in S$, let $\phi_t$ be evolution by the classical Hamiltonian flow of $p$, induced by the Hamiltonian vector field $X_p=\dd_{\xi}p\dd_x-\dd_xp\dd_{\xi}$. On the quantum side, a particle evolves by the Schr\"odinger equation $$hD_t\psi=P\psi,$$ so $\psi(t)=e^{-\frac{i}{h}tP}\psi(0)$ and observables $A$ evolve in the Heisenberg picture by $A(t)=e^{\frac{i}{h}tP}Ae^{-\frac{i}{h}tP}$. Egorov's theorem states that this evolved operator is equal, up to order $h$, to the quantization of the classically observed symbol of $A$, for $t$ less than \emph{Ehrenfest time} $t_E\sim\log\frac{1}{h}$. Specifically, we have that if $\Gamma$ is the Lyapunov exponent of $\phi_t$ and \begin{equation}\label{eq:Ehrenfest_time}|t|\le T+\frac{\gd}{\Gamma+\eps}\log\frac{1}{h}\end{equation} for some $\delta<\frac{1}{2}$ and $\eps>0$, then
\begin{equation}\label{eq:QuantEgorov}
 e^{\frac{i}{h}tP}\Op_h(a)e^{-\frac{i}{h}tP}-\Op_h(\phi_t^*a) = \Op_h(r)
\end{equation}
for $a\in S$, and $r\in S_{\gd}$ has all seminorms of order $O(h^{2-3\gd})$.
In cases where $\Gamma=0$, such as a completely integrable classical system, some stronger results are available, see \cite{BR}.

\subsection{Fourier Integral Operators} \label{s:Fourier integral operators}

As a generalization of the operators $e^{-\frac{i}{h}tP}$ in Egorov's theorem, we briefly discuss Fourier integral operators. One may view the operator $e^{-\frac{i}{h}P}$ as quantizing a transformation (or equivalently, a ``change of variables") on phase space; namely the symplectomorphism on $T^*\R^d$ given by $(x,\xi)\mapsto\phi_t(x,\xi)$. It is then natural to generalize the quantization to arbitrary symplectomorphisms. This abstraction will be rewarded when we work on the quantized torus in \S \ref{ss:quantizedtorus}, on which the most natural classical transformations do not arise from a flow at all. Let $\phi:T^*\R^d\to T^*\R^d$ be any symplectomorphism. We consider only unitary Fourier integral operators here, and only consider their action in the case $\gd=0$. For the purposes of this paper, define a \emph{Fourier integral operator quantizing $\phi$} in the following way.

\begin{defi}[quantized flow]
For $\phi : T^* \mathbb {R}^d \to T^* \mathbb R^d$, the Fourier integral operator quantizing $\phi$ is the operator $U:L^2(\R^d)\to L^2(\R^d)$ such that the Egorov-type relation
\begin{equation*}\label{eq:egorovtype}
 U^{-1}\Op_h(a)U=\Op_h(b)
\end{equation*}
holds for any $a\in S$ with $b=\phi^*a+O_{S}(h)$. 
\end{defi}
In particular (using these definitions), Egorov's theorem is precisely the statement that $e^{\frac{i}{h}tP}$ is a one-parameter family of Fourier integral operators, which quantize the classical Hamiltonian flow $\phi_t$.

Perhaps the most important Fourier integral operators are the \emph{metaplectic operators} \cite[Chapter 11.3]{Zworski}, which is the special case when $\phi$ is a linear symplectomorphism. In this case for $M\in\Sp(2d,\R)$, we write the corresponding metaplectic operator as $\wh{M}$. The metaplectic operators are unitary and benefit from an exact version of (\ref{eq:egorovtype}), ie they satisfy $$\wh{M}^{-1}\Op_h(a)\wh{M}=\Op_h(M^*a)$$ for all $a\in S$, with no error. We may remark that while the metaplectic operators define a projective unitary representation of $\Sp(2d,\R)$, they are in fact a unitary representation of its double cover, which is known as the \emph{metaplectic group}. This is analogous to projective representations of rotation operators on spinor spaces and leads similarly to a non-canonical choice of phase for a metaplectic operator. As we will only encounter metaplectic operators when we are conjugating by them, this phase vanishes and we need not concern ourselves with it.

\subsection{Semiclassical Defect Measures}

We now discuss semiclassical defect measures, which give a quantitative answer to where in phase space a quantum particle ``is" in the semiclassical limit as $h\to0$.

\begin{defi}[defect measure] \label{def:defect measure}
We say an $h$-dependent set of density operators $\rho_h$ on $L^2(\R^d)$, has \emph{defect measure} $\mu$ if for all $a\in\cinf_0(T^*\R^d)$,
\begin{equation}\label{eq:DM}
 \lim_{h\to0}\tr\left(\rho_h\Op_h(a)\right)=\int_{T^*\R^d}a\ddd \mu.
\end{equation} If $\psi_h\in L^2(\R^d)$, we say it has defect measure $\mu$ if the density operator $\ket{\psi_h}\bra{\psi_h}$ has defect measure $\mu$.
\end{defi}

Having a semiclassical defect measure is a very special property of a sequence, corresponding to being ``essentially classical" as $h\to0$. Nevertheless, such a limit necessarily exists along a subsequence of $h$'s, as the following proposition states.
\begin{prop}\label{p:DMExistence}
Let $\rho_h$ be an $h$-dependent family of density operators on $L^2(\R^d)$. Then there is a nonnegative Radon measure $\mu$ on $T^*\R^d$ and a sequence $h_j\to0$ such that
\begin{equation*}
 \lim_{j\to\infty}\tr\left(\rho_{h_j}\Op_{h_j}(a)\right)=\int_{T^*\R^d}a\ddd \mu
\end{equation*}
for all $a \in C_0^\infty (T^* \R^d )$.
\end{prop}
The proof of Proposition \ref{p:DMExistence} is essentially the same as \cite[Theorem 5.2]{Zworski}, which proves it in the special case of pure states.

In many natural cases, we may take the limit in (\ref{eq:DM}) to be uniform in the symbol $a$, in which case we refer to $\mu$ as a \emph{uniform defect measure}. For a precise characterization of aforesaid measures, see Proposition \ref{p:UDF_Equivalence}. In particular, in this paper, we will work on the quantized torus $H_N$ (see \S \ref{ss:quantizedtorus}) where all defect measures are uniform.

For semiclassical pseudodifferential operators on the real line, there are standard classes of states with known defect measures \cite[Proposition 3.3]{BFF} which we restate in Proposition \ref{p:Rdcoherent}. For a full proof (in French) see \cite[\S 3]{Lions1993} and a proof for less general coherent states see \cite[Chapter 5.1]{Zworski}.

\begin{prop}\label{p:Rdcoherent}
Fix $g\in \mathcal{S}(\mathbb R^d)$ with $\| g \|_{L^2} = 1$, define $$\psi_h(x): =h^{-\frac{d\gb}{2}}g\left(\frac{x-x_0}{h^{\gb}}\right)\exp\left(i\dfrac{x\cdot\xi_0}{h}\right),$$ and let $\rho_h=\ket{\psi_h}\bra{\psi_h}.$ Then $\rho_h$ has defect measure $\mu$, given as follows: 
\begin{itemize}
 \item If $\beta=0$, $\mu=\gd_{\xi=\xi_0}|g(x-x_0)|^2\ddd x$,
 \item If $\beta=1$, $\mu=\gd_{x=x_0}|\hat{g}(\xi-\xi_0)|^2\ddd \xi$,
 \item If $\beta\in(0,1)$, $\mu=\gd_{x=x_0,\xi=\xi_0}$.
\end{itemize}
Here $\hat{g}$ is Fourier transform of $g$.
\end{prop} In the last case of Proposition \ref{p:Rdcoherent}, the states $\psi_h$ are called \emph{generalized coherent states}.

In fact, as we will be working with $S_{\delta}$ classes and must control the error of our estimates, we will need a slightly stronger definition to accommodate $h$-dependent symbols.
\begin{defi}[uniform defect measure]
    If $\rho_h$ has defect measure $\mu$, $\gd\in[0,\frac{1}{2})$, and $\gt>0$, then we say $\mu$ is a \emph{uniform $(\gd,\gt)$-defect measure} provided \begin{equation}\label{eq:delta_DM}\left|\tr(\rho_h\Op_h(a))-\int_{T^*\R^d}a \ddd \mu\right|=O(h^{\gt})\end{equation} holds for all $a=a_h\in \cinf_0(T^*\R^d)\cap S_{\delta}$ with constant depending only on the symbol norms of $a$. 
\end{defi}
We easily see that if $\mu$ is a $(\gd,\gt)$-defect measure, it is necessarily a $(\gd',\gt')$-defect measure for all $\gd'\le\gd$ and $\gt'\le\gt$. To address a potential point of confusion, note that in (\ref{eq:delta_DM}) we are enlarging the space of symbols, $a$, used in \eqref{eq:DM}. These symbols are allowed to depend on $h$ with derivatives growing in $h$, provided they do not grow too fast, and their support may not be uniformly bounded in $h$.

We can now state a refinement of Proposition \ref{p:Rdcoherent} using this new terminology.

\begin{prop}\label{p:Rdcoherent_strong}
Let $\rho_h$ be as in Proposition \ref{p:Rdcoherent}. Then the defect measures $\mu$ given in Proposition \ref{p:Rdcoherent} are uniform $(\gd,\gt)$-defect measures with $\gd,\gt$ given as follows:
\begin{itemize}
 \item If $\beta\in\{0,1\}$, then $\gd<\frac{1}{2}$ and $\gt<1-2\gd$,
 \item If $\beta\in(0,1)$, then $\gd<\min(\gb,1-\gb)$ and $\gt<\min(\gb-\gd,1-\gb-\gd)$.
\end{itemize}
\end{prop}
The proof of Proposition \ref{p:Rdcoherent_strong} is straightforward and given in Appendix \ref{aa:coherent_Rd}.

\subsection{The Quantized Torus}\label{ss:quantizedtorus}

We now give an introduction to the quantized torus, and the pseudodifferential calculus on the Hilbert space of corresponding quantum states.

For each $N\in \N$, define
\begin{align}
    H_N:=\mathrm{span}\left\{\frac{1}{N^{\frac{d}{2}}}\sum_{k\in\Z^d}\delta_{x=k+\frac{n}{N}}:n\in\{0,1,\dots,N-1\}^d\right\}. \label{eq:quantizedTorus}
\end{align}
We may remark that $H_N$ consists precisely of the elements of $\Sc'(\R^d)$ which are periodic in physical and Fourier space under the semiclassical Fourier transform $\cF_h$, given by $$\cF_hu(\xi)=\frac{1}{(2\pi h)^{\frac{d}{2}}}\int e^{\frac{i}{h}x\cdot\xi}u(x)\ddd x$$ where $h=(2\pi N)^{-1}$. For simplicity define
\begin{align}
 Q_n=\frac{1}{N^{\frac{d}{2}}}\sum_{k\in\Z^d}\delta_{x=k+\frac{n}{N}}, \label{eq:Q_n def}
\end{align}
and define a Hilbert space structure on $H_N$ such that $(Q_n)$ is an orthonormal basis of $H_N$. Though it is not necessary here, we may note that this is an elementary example of the general geometric Toeplitz quantization studied in \cite{Deleporte}.

We now give an overview of the pseudodifferential calculus on $H_N$, which also may be found in \cite{CZ},
\cite{DJ}, \cite{Schenck}. Let $a\in \cinf(\T^{2d})$. Intuitively, $a$ corresponds to a classical observable, with $\T^{2d}$ playing the role of phase space. We see that $a$ lifts to a doubly-periodic function $\tilde{a}$ on $T^*\R^d$, which is in the symbol class $S$. By \cite[\S 2.3]{CZ} $\Op_Na$ maps $H_N$ (as a subset of $\Sc'(\R^d)$) to itself, so we may define \begin{equation}\label{eq:TN_Op}
\Op_Na:=\Op_h\tilde a|_{H_N}:H_N\to H_N 
\end{equation} which is also given in coordinates \cite[Lemmma 2.4]{CZ} by $$\Op_N(a)Q_j=\sum_{m=0}^{N-1}A_{mj}Q_m$$ with $$A_{mj}=\sum_{k,l\in\Z}\hat{a}(k,j-m-lN)(-1)^{kl}e^{\pi i\frac{(j+m)k}{N}},$$ where $\hat{a}$ is the (non-semiclassical) Fourier series in both variables. Furthermore, we may define an analogue of $S_{\gd}$ symbols as $$S_{\gd}(\T^{2d}):=\{a\in \cinf(\T^{2d}):\forall \alpha,\beta\in\Z^d,\,|\dd_x^{\ga}\dd_{\xi}^{\gb}|\le C_{\ga\gb}N^{\gd|\ga+\gb|}\}$$ from which the obvious generalization of (\ref{eq:TN_Op}) can be defined.

Most results from pseudodifferential calculus carry over to $H_N$, see \cite{CZ,DJ} for details. In particular, the formulas for compositions and Egorov's theorem are the same and follow immediately, and we have a version of the Calder\'on-Vaillancourt theorem \cite[Proposition 2.7]{CZ} that says $$\|\Op_N(a)\|_{H_N}\le\|a\|_{L^{\infty}(\T^{2d})}+o(1),$$ where $\|\cdot\|_{H_N}$ is the norm given by the Hilbert space structure. In fact, one can directly bound $\Op_N (a)$ by the corresponding map on $L^2(\R^d)$: \begin{equation}\label{eq:HN_L2_bound}\|\Op_N(a)\|_{H_N}\le\|\Op_h(\tilde{a})\|_{L^2(\R^d)},\end{equation} see \cite[\S 2.2.3]{DJ} for an explanation using the direct integral decomposition of $L^2(\R^d)$.

A particular class of metaplectic operators can also be defined on $H_N$. Let $M\in \Sp(2d,\Z)$ be an integer symplectic matrix, and let $N$ be even. Define $\wh M_N=\wh M|_{H_N}$, which can be shown \cite[\S 2.2.4]{DJ} to be a unitary map on $H_N$. We then have an analogous exact Egorov's theorem:
\begin{equation}\label{eq:MNEgorov}
 \Op_N(a)\wh{M}_N=\wh{M}_N\Op_N(a\circ M).
\end{equation}

The definition of a defect measure on $H_N$ is also identical. Given an $N$-dependent sequence of density operators on $H_N$, we say they have \emph{defect measure} $\mu$ if for every $a\in\cinf(\T^{2d})$,
\begin{equation*}
 \lim_{N\to\infty}\tr\left(\rho_N\Op_N(a)\right)=\int_{T^*\R^d}a\ddd \mu.
\end{equation*}

The proof of Proposition \ref{p:UDF_Equivalence} shows that every defect measure on the torus is uniform. In particular, we have the following result to be an immediate consequence.

\begin{prop}\label{p:TorusDMUniform}
If $\rho_N$ are density operators on $H_N$ with defect measure $\mu$, then $\mu$ is a probability measure, and for every $\eps>0$ there is an $h_0$ such that if $h<h_0$ and $\|a^{(\ga)}\|\le 1$ for all $|\ga|<K(d)$, then $$\left|\tr_{H_N}(\rho_N\Op_N(a))-\int_{\T^{2d}}a\ddd \mu\right|\le \eps.$$

\end{prop}

The notion of a $(\gd,\gt)$-defect measure (\ref{eq:delta_DM}) also carries over to the quantized torus case, where the symbols in this case are in $S_{\gd}(\T^{2d})$.

Proposition \ref{p:Rdcoherent}, which gives the defect measures of generalized coherent states, has a natural analog on the quantized torus. We do not claim these values of $\gd,\gt$ to be sharp.

\begin{prop}\label{prop_cs_beta0}
Given $g \in C^\infty(\T^d ; \C)$ such that $\| g \| _{L^2 (\T^d ) } = 1$, and $\xi_0 \in \T^d$, define
\begin{align*}
 \psi_N := \sum_{k \in (\Z / N \Z )^d } N^{-\frac{d}{2}}g\left (\frac{k}{N} \right)e^{2\pi ik\cdot\xi_0} Q_k,
\end{align*}
and set $\rho_N:=C_N\ket{\psi_N}\bra{\psi_N}$ with $C_N:=\|\psi_N\|_{H_N}^{-2}$ (for $Q_k$ defined in \eqref{eq:Q_n def}). Then $\rho_N$ has the semiclassical defect measure $\mu=\gd_{\xi=\xi_0}|g(x)|^2\ddd x$. In particular, $\mu$ is a $(\gd,\gt)$-defect measure for $\gd<(d+2)^{-1}$, $\gt<1-(d+2)\gd.$
\end{prop}

To prepare for the next proposition, we define the map $\Pi_N:\Sc(\R^d)\to H_N$ by
\begin{equation}
 (\Pi_Nu)_n=\frac{1}{N^{\frac{d}{2}}}\sum_{k\in\Z^d}u\left(k+\frac{n}{N}\right)Q_n. \label{eq:PiDefinition}
\end{equation}
\begin{prop}\label{prop:cs_beta_int}
Given $\beta \in (0,1)$ and $g\in \Sc(\R^d)$, define
\begin{align*}
 g_N(x):=N^{\frac{d\gb}{2}}g(N^{\beta}(x-x_0))e^{2\pi iNx\cdot\xi_0}.
\end{align*}
Next, define $\psi_N:=\Pi_Ng_N$, and $\rho_N:=C_N'\ket{\psi_N}\bra{\psi_N}$, where $C_N':=\|\psi_N\|_{H_N}^{-2}$. Then $\rho_N$ has the semiclassical defect measure $\mu=\gd_{x=x_0,\xi=\xi_0}$, which is a $(\gd,\gt)$-defect measure for $\gd<\min(\gb,\frac{1-\gb}{d+1})$, $\gt<\min(\gb-\gd,1-\gb-(d+1)\gd)$.
\end{prop}

\begin{rem}\label{r:normalization}
We can easily see that $C_N, C_N'=1+o(1)$. Indeed, by Riemann sums:
\begin{equation*}
 C_N^{-1}=N^{-d}\sum_{k\in(\Z/N\Z)^d}\left|g\left(\frac{k}{N}\right)\right|^2=1+O(N^{-1}),
\end{equation*}
so $C_N=1+O(N^{-1})$. Similarly,
\begin{equation*}
\begin{gathered}
 (C_N')^{-1}=N^{-d(\gb-1)}\sum_{m\in(\Z/N\Z)^d}\left|g\left(N^{\gb}\left(\frac{m}{N}-x_0\right)\right)\right|^2\\=N^{-d}\sum_{m\in(\Z/N\Z)^d}N^{d\gb}\left|g\left(N^{\gb}\left(\frac{m}{N}-x_0\right)\right)\right|^2=1+O(N^{-(1-\gb)}),
 \end{gathered}
\end{equation*}
so $C_N'=1+O(N^{-(1-\gb)})$.

In particular, this shows that the constants $C_N$, $C_N'$ in Propositions \ref{prop_cs_beta0} and \ref{prop:cs_beta_int} can actually be taken to be $1$ without affecting the value of the defect measure, or even the possible values of $(\gd,\gt)$. They are only included so the density operators match our physical motivation of states with norm exactly $1$.
\end{rem}

The proofs of Propositions \ref{prop_cs_beta0} and \ref{prop:cs_beta_int} are given in Appendix \ref{aa:coherent_torus}.

\subsection{Measurement Procedure} \label{s:measurement procedure}

We provide a brief background on the quantum mechanics used in this paper. For a more comprehensive overview, see \cite{BG} \cite{Holevo}. To model particles in quantum mechanics, we begin by fixing a Hilbert space $H$, known as the \emph{state space}. Particles in our model are described by \emph{density operators}.

\begin{defi}[density operator]\label{def:desnity operator}
A density operator on $H$ is a symmetric, positive semidefinite, trace-class operator on $H$ with trace $1$.
\end{defi}

Quantum particles in this paper are measured to have some value in $\Omega$ (a locally compact metric space) by a type of quantum instrument called \emph{Kraus operators}.

\begin{defi}[Kraus operators] \label{def:Kraus operators}
Suppose $\Omega$ is a locally compact metric space with Borel measure $\nu$. And suppose for $q\in \Omega$, $A_q$ is a bounded operator on $H$, $q\mapsto A_q$ is measurable, and
\begin{align*}
    \int_\Omega A_q^* A_q \ddd \nu (q) = I.
\end{align*}
Then we call $A_q$'s \emph{Kraus operators}.
\end{defi}

Given a quantum particle described by a density operator $\rho$ and Kraus operators $A_q$ on $\Omega$, the probability of measuring $\rho$ in $F\subset \Omega$ (a measurable set) is:
\begin{align*}
    \tr \left( \int_F A_q \rho A_q^* \ddd \nu (q)  \right).
\end{align*}
After measurement within $F$, the state changes to:
\begin{align*}
    \rho(F) : = \frac{\int_F A_q \rho A_q^* \ddd \nu (q)  }{ \tr \left( \int_F A_q \rho A_q^* \ddd \nu (q)  \right)}.
\end{align*}
For each measurable set $F \subset \Omega$, we can define $\mathcal{I} (F) := \rho \mapsto \int_F A_q \rho A_q^* \ddd \nu (q) $. In this case, $\mathcal{I}$ is an example of a quantum instrument. 

This paper considers a specific class of Kraus operators. The Hilbert space we work on is the quantized torus $H_N$ (recall \S \ref{ss:quantizedtorus}). We fix a compact metric space $\Omega$ with finite Borel measure $\nu (q)$, and let $f_q\in \cinf(\T^{2d})$ satisfy 
\begin{equation}\label{eq:f_q satisfy condition}
    \int_{\Omega}\Op_N(f_q)^*\Op_N(f_q) \ddd \nu(q)=Id,
\end{equation}
where $\Op_N$ is defined in \eqref{eq:TN_Op}. Additionally, assume that $f_q$ are in $S$ (recall Definition \ref{def:symbol class}) uniformly in $q$, meaning $\|\dd^{\ga}f_q\|_{L^{\infty}}\le C_{\ga}$ with constant independent of $q$. These $\Op_N(f_q)$ are the Kraus operators we will consider and we will refer to $f_q$ as the symbols of the Kraus operators.

As a natural example of the above construction, the reader may consider the example where $\Omega=\T^d$, $\ddd \nu(q)=\ddd q$, and $f_q(x,\xi)=f(x-q)$, where $\int_{\T^d}|f(x)|^2\ddd x=1$. Then conjugation by $\Op_N(f_q)$ corresponds intuitively to an ``approximate position measurement" (see also \cite{BCFFS}), where the approximation approaches a (nonexistent) ``exact" position measurement as $f$ approaches a delta function.

\subsection{Evolution Procedure}
This paper models quantum particles which are repeatedly measured and evolved. The previous section discussed the measurement procedure, here we discuss the evolution procedure.

We consider two different evolution procedures. Particles will either be evolved by (1) a metaplectic operator or (2) exponentiation of a quantization of a Hamiltonian on the torus. In the first case, we fix $M\in\Sp(2d,\Z)$ and let $\wh{M}_N$ be the corresponding metaplectic operator acting on $H_N$ (recall \S \ref{s:Fourier integral operators}). In the second case, we fix $p\in\cinf(\T^{2d})$, and let $P=\Op_N(p)$ (recall \eqref{eq:TN_Op}).

Define an operator $\Phi_{N,q}^*$ on density operators given by \begin{equation}\label{eq:defPhi}\Phi_{N,q}^*(\rho):=U\Op_N(f_q)\rho\Op_N(f_q)^*U^{-1}
\end{equation} where $U$ is a unitary operator given either by $U=\wh{M}_N$ or $U=e^{-2\pi iNP}$.
We remark that $\Phi^*_{N,q}(\rho)$ is $U\tilde{\rho}(q)U^{-1}$ where $\tilde{\rho}(q)=\Op_N(f_q)\rho\Op_N(f_q)^*$ is an un-normalized a posteriori state of the instrument determined by Kraus operators $\Op_N(f_q)$. We stress again that $\Phi^*_{N,q}(\rho)$ is not a density operator, as it is not normalized; the trace will be taken at the end and interpreted as the probability density of the trajectory $(q_0,\dots, q_n)$.

The evolution of the quantum system by $U$ corresponds to a classical dynamical system $\phi_t : \T^{2d} \to \T^{2d}$. In the case of evolving by the metaplectic operator $\widehat{M_N}$, $\phi_t : = M^t$ (where $t$ takes integer values). In the case of evolving by $e^{-2\pi i N t\Op_N(p)}$, the corresponding classical dynamical system is $\phi_t : = \exp(tH_p)$ where $H_p$ is the Hamiltonian vector field generated by $p$.

The quantitative convergence rates of the measures on the quantum and classical trajectories depend on the Lyapunov exponents (recall Definition \ref{def:Lyapunov exponent}) of these corresponding classical dynamical systems. A system with a larger Lyapunov exponent is more exponentially sensitive to initial conditions, and hence seen as ``more chaotic." In particular, completely integrable systems have Lyapunov exponents equal to zero.

Lyapunov exponents appear crucially in the following bound on symbol seminorms of a time-evolved symbol.

\begin{prop}\label{p:Evolved_Symbol_Bound} 
Suppose $\phi_t$ is a smooth dynamical system on $\T^{2d}$ where $t$ either takes values in $\Z_{\ge 0}$ or $\mathbb R_{\ge 0}$ (in which case we assume $\phi_t$ is smooth in time), with Lyapunov exponent $\G$. Suppose $a\in C^\infty (\T^{2d})$, and define $a_t=a\circ\phi_t$. Then for $\ga\in\N^{2d}$, $\eps>0$, $$\|\dd^{\ga}a_t\|_{L^{\infty}}\le C_{\ga,\eps}e^{(\G+\eps)|\alpha|t}\sup_{|\gb|\le|\ga|}\|\dd^{\gb}a\|_{L^{\infty}}.$$ In particular if $a\in S$ and $t\le\frac{\gd}{\G+\eps}\log\frac{1}{h}$, then $a_t\in S_{\delta}$ with seminorms independent of $t$.\end{prop}
A version of Proposition \ref{p:Evolved_Symbol_Bound} on manifolds is discussed in \cite[\S 5.2]{AN}, \cite[Appendix C]{dyatlov2014}, and \cite[\S 5.2]{dyatlov2022}. We present a proof of Proposition \ref{p:Evolved_Symbol_Bound} in Appendix \ref{s:derivative of flow}, adapting the arguments of the mentioned references to the simpler case of the torus.

\section{Main Result and Proof}\label{s:theorems}

We are now able to provide a quantitative proof of our result.

\begin{thm}\label{t:torusquant}
Suppose $H_N$ is the quantized torus of dimension $2d$, $\rho=\rho_N:H_N\to H_N$ is a set of density operators depending on $N$ with semiclassical defect measure $\mu$ on $\T^{2d}$, $\Omega$ is a compact metric space with finite Borel measure $\nu$, $f_q \in C^\infty(\mathbb T^{2d})$ satisfy \eqref{eq:f_q satisfy condition},  $P_{\mu}^{(n)}$ is defined by (\ref{eq:classicalmeasure}), and $\bP_{N,\rho}^{(n)}$ is defined by (\ref{eq:quantummeasure}), with $\Phi_{N,q}^{*}$ defined by $(\ref{eq:defPhi})$. Then for every $n\in\N$ we have as $N\to\infty$ the convergence in total variation 
\begin{equation*}
 \bP_{N,\rho}^{(n)}\xrightarrow{N\to\infty}P_{\mu}^{(n)}.
\end{equation*}

Furthermore, we have the quantitative estimate that if $\mu$ is a $(\gd,\gt)$-defect measure as defined in (\ref{eq:delta_DM}) and $\eps>0$, then there are $C_0, C_1$ independent of $n$ such that if \begin{equation}\label{eq:ehrenfest_condition}n<\frac{\gd\log(N/C_0)}{\G+\eps},\end{equation} 
then
\begin{align}\label{eq:thm_thing}
 \left\|\bP_{N,\rho}^{(n)}-P_{\mu}^{(n)}\right\|_{TV}\le C_1N^{-\min(1-2\gd,\theta)}
\end{align}
where $\Gamma$ is the Lyapunov exponent of $M$ (metaplectic case) or $\phi_t$ (Hamiltonian evolution case). \end{thm}
\begin{rem}\label{r:Ehrenfest}
The condition (\ref{eq:ehrenfest_condition}) can be recognized as the \emph{Ehrenfest time} given in (\ref{eq:Ehrenfest_time}).
\end{rem}

\begin{proof}[Proof of Theorem \ref{t:torusquant}]
This proof goes along the lines of \cite[Theorem 3.1]{BFF}, but gives quantitative bounds. The version for $U=e^{-2\pi iNP}$ requires the quantitative version of Egorov's theorem given in (\ref{eq:QuantEgorov}).

We first consider when $U=\wh{M}_N.$ We have

\begin{align}
\label{eq:proof}
 \begin{split}
 &\ddd \bP^{(n)}_{N,\rho}(q_0,\dots,q_n)= \tr_{H_N} \left ( \prod_{i=n}^0 \left(\wh{M}_N \Op_N (f_{q_i}) \right) \rho \prod_{i=0}^n \left( \Op_N (f_{q_i})^* \wh{M}_N^{-1} \right) \right) \ddd \nu (q)\\
 &= \tr _{H_N} \left( \prod_{i=0}^n \left( \Op_N (f_{q_i})^* \wh{M}_N^{-1} \right) \prod _{i=n}^0 \left ( \wh M _N \Op_N (f_{q_i})\right) \rho \right) \ddd \nu (q)\\
 &= \tr _{H_N} \left( \prod_{i=0}^n \left ( \Op_N (f_{q_i} \circ M^i )^* \right) \prod _{i=n}^0 \left( \Op_N (f_{q_i} \circ M^i ) \right) \rho \right)\ddd \nu (q)\\
 &= \left(\tr_{H_N}\left(\Op_N(|f_{q_0}|^2\cdots|f_{q_n}\circ M^n|^2)\rho\right)+\poly(n)O(N^{-(1-2\gd)})\right)\ddd \nu(q)\\
 &=\left(\int_{\T^{2d}}|f_{q_0}(\zeta)|^2\cdots|f_{q_n}(M^n\zeta)|^2\ddd \mu+\poly\log(N)(O(N^{-(1-2\gd)})+O(N^{-\theta}))\right)\ddd \nu(q).
 \end{split}
\end{align}

The first three equalities follow from the definition of $\Phi_{N,q_j}^*$, commutativity of the trace, the exact Egorov property (\ref{eq:MNEgorov}) of $\wh{M}_N$, and the unitarity of $\wh{M}_N^*$. We then used Proposition \ref{p:Evolved_Symbol_Bound} to see that for each $i$, $f \circ M^i \in S_\delta$, for 
\begin{align*}
 \delta \ge \frac{(\Gamma + \e ) n }{\log (N)}
\end{align*}
so that we could use the composition rule for pseudodifferential operators (\ref{eq:n_comp_rule}) together with (\ref{eq:HN_L2_bound}) to get the fourth equality. We lastly used that $\rho$ has $(\gd,\gt)$-defect measure $\mu$ to get the last equality. Note that we may remove the $\poly\log(N)$ term by adjusting $\eps$ and replacing $\gd$ with a slightly smaller $\gd'$.

On the other hand, 
\begin{align*}
 \ddd P^{(n)}_{\mu}=\int_{\T^{2d}}|f_{q_0}(\zeta)|^2\cdots|f_{q_n}(M^n\zeta)|^2\ddd \mu\ddd \nu(q).
\end{align*}
We may now integrate in $q$ and apply the dominated convergence theorem (using that the $\|f_q \| $ are uniformly bounded in $q$) to get the theorem for the case that $U=\wh{M}_N$.

When $U=e^{-2\pi i NP}$, the proof is the same but uses the quantitative version of Egorov's theorem (\ref{eq:QuantEgorov}) on the second line of (\ref{eq:proof}), to give
\begin{align*}
    &\tr _{H_N} \left( \prod_{i=0}^n \left( \Op_N (f_{q_i})^* e^{2\pi i NP} \right) \prod _{i=n}^0 \left ( e^{-2\pi i NP} \Op_N (f_{q_i})\right) \rho \right) \ddd \nu (q)\\
    & \quad =\left(\tr _{H_N} \left( \prod_{i=0}^n \left ( \Op_N (f_{q_i} \circ \phi_i )^* \right) \prod _{i=n}^0 \left( \Op_N (f_{q_i} \circ \phi_i ) \right) \rho \right)\right)\ddd \nu (q)\\
    & \qquad \qquad \qquad \qquad \qquad \qquad + \left( \poly\log(N)O(N^{-(2-3\gd)})\right)\ddd \nu (q).
\end{align*}
The $O(N^{-(2-3\gd)})$ error is dominated by the $O(N^{-(1-2\gd)})$ error later on in (\ref{eq:proof}), and the remainder of the proof is the same.
\end{proof}

\section{Numerical Illustrations}\label{s:numerics}

Here we present numerical simulations to illustrate the results of this paper. In these numerics, we simulate the evolution of a state on the quantization of the 2-dimensional torus. The measurement procedure will use Kraus operators $f_q (x,\xi) = f(x - q)$, with 
\begin{align}\label{eq:f_def}
 f(x) = c \left ( \sum_{k\in \mathbb Z} \exp\left (\frac{-(x - 2\pi k)^2}{2 \sigma^2}\right)\right)^{1/2}
\end{align}
where $\sigma > 0$ is a parameter to be chosen\footnote{$\sigma$ can be viewed as the measurement precision. The smaller the value of $\sigma$, the more precise.} and $c>0$ is such that $\int_\mathbb T |f|^2 = 1$. At each time step, after observation, we evolve the particle with the metaplectic operator $\wh M_N$ (as defined in \S \ref{ss:quantizedtorus}). For now, we let $M$ be Arnold's cat map,
\begin{align} 
 M = \mat{2 & 1 \\ 1 & 1}. \label{eq:cat_map_matrix}
\end{align}
We also choose the initial state to be $\rho = \rho (N) = | \psi _N \rangle \langle \psi _N |$ with
\begin{align}
 \psi_N = \Pi_N \left( \frac{N^{\frac{\beta}{2}}}{\sigma' \sqrt{2\pi}} \exp \left ( \frac{ -((x - x_0 ) N^\beta )^2 }{2 (\sigma') ^2} + 2\pi i N x p_0 \right)\right) \label{eq:initial state}
\end{align}
with $\Pi_N$ defined in \eqref{eq:PiDefinition} and some fixed $x_0, p_0 \in [0,1]$ and $\sigma' >0$. By Proposition \ref{prop:cs_beta_int}, $\rho$ has semiclassical defect measure $\delta _{x = x_0,\xi = p_0}$ when $\beta \in (0,1)$.

Let $Q_i$ be the observed location of the state at time $i$. These $Q_i$'s are random variables, and their joint distribution is described by \eqref{eq:quantummeasure}. To numerically simulate this distribution, we note that formally $\bP\left (\bigcap _{i=0}^n \set{Q_i = q_i} \right )$ can be written
\begin{align*}
\bP\left (Q_n = q_n \; \middle| \;\bigcap _{i=0}^{n-1} \set{Q_i = q_i}\right )\bP\left (Q_{n-1} = q_{n-1} \;\middle| \; \bigcap _{i=0}^{n-2} \set{Q_i = q_i}\right )\cdots \bP(Q_0 = q_0).
\end{align*} 
Therefore to sample $Q_i$ from this distribution, we first sample $Q_0$ with law $\bP(Q_0 = q_0) \ddd \nu(q_0)$, and get a value, which we call $\bar q_0$. Then we sample $Q_1$ with law $$\bP \left ( Q_1 = q_1 \; \middle| \; Q_0 = \bar q_0 \right ) \ddd \nu(q_1).$$ We continue this process to get values of $\bar q_i$.

For our choice of initial state and Kraus operators, the probability density functions of these distributions greatly simplify. Working in the basis of $H_N$ described in \S \ref{ss:quantizedtorus}, we write $|\psi_N \rangle$ as a vector $v(0) \in \C^N$ (the argument of $v$ will denote the time at which the state is to be observed). It is a straightforward computation that, with respect to this basis, $$\Op_N (f_q) = \diag(f(x_1 - q) , f(x_2 - q), \dots , f(x_N - q) )$$ where $x_i = i/N$. To ease notation, let $\wh M_N := U$ and $\Op_N(f_{q}) := A_q$. Then, by \eqref{eq:quantummeasure},
\begin{align*}
\bP(Q_0 = q_0 )& = \tr(U A_{q_0 } v(0) \bar v(0) ^t A_{q_0} U^*) \\
&= \sum_k (A_{q_0} )_{k,k}^2 v(0)_k^2 = \sum_k f(x_k - q_0)^2 v(0)_k^2.
\end{align*}
Observe that this is the convolution of the probability density functions $f(x)^2$ and $\sum_k 1_{x_k} v(0)_k^2$ (call random variables with these distributions $F$ and $V(0)$ respectively). Here $1_{x_k}$ is $1$ at $x_k$ and zero everywhere else. Therefore to sample $Q_0$, we sample from the random variable $F + V(0)$ to get $\bar q_0$. 

To sample $Q_i$, for $i > 1$, we proceed in a similar manner. Assuming we have already computed that $Q_j = \bar q_j$ for $j <i$, we must sample a random variable with distribution
\begin{equation}\label{eq:177}
\begin{split}
\bP&\left (Q_i = q_i \; \middle| \; \bigcap _{j=1}^{i-1} Q_j = \bar q_j \right ) \ddd \nu(q_i) \\
&= \bP \left ( \bigcap _{j=1}^{i-1} Q_j = \bar q_j \right)^{-1} \tr\left ( A_{q_i} \prod _{j=1}^{i} (U A_{\bar q_{i-j} } ) v(0) v(0)^t \prod _{j=0}^{i-1} (A_{\bar q_j} U^*) A_{q_i}\right )\ddd \nu(q_i). 
\end{split}
\end{equation}
Let:
\begin{align*}
v(i) = \prod _{j=1}^i (U A_{\bar q_{i-j}}) v(0),
\end{align*}
so that \eqref{eq:177} is $c\sum _ k f(x_k - q_i)^2 v(i)_j^2$, with $c^{-1}= \sum_k v(i)_k^2$. Like before, let $V(i)$ be a random variable with probability density function $\sum_k v(i) _k^2 1_{x_k}$. Therefore, to sample $Q_i$, sample from $V(i) + F$. This proves the following proposition.

\begin{prop}
Suppose $\rho = |\psi_N \rangle \langle \psi_N | $  is an initial state (as defined by \eqref{eq:initial state}), $f(x-q)$ are the symbols of Kraus operators (given by \eqref{eq:f_def}), and $n > 0$. Then a random trajectory of a quantum particle with distribution $\mathbb P^{(n)} _{N,\rho}$, described by \eqref{eq:quantummeasure}, can be sampled by applying Algorithm \ref{qmap} by storing the values $q_0, q_1,\dots ,q_n$ (representing the observed locations of the particle at each instance of time).

\begin{algorithm}
\caption{Quantum Particle Evolution}\label{qmap}
\begin{algorithmic}[1]
\State $\textbf{v} \gets |\psi _N \rangle$ 
\While{$i \le n$}
\State $q_i \gets $ sampled random number from distribution $\sum_k 1_{x_k} \textbf{v}_k^2$
\State $g \gets $ sampled random number from Gaussian distribution of mean $0$, variance $\sigma^2$
\State $q_i \gets q_i + g \pmod 1$
\State $\textbf{v} \gets \diag (f(x_1 - q_i), \dots, f(x_N - q_i) ) \textbf{v}$
\State $\textbf{v} \gets \hat M_N \textbf{v}$
\State $\textbf{v} \gets \norm{ \textbf{v}}^{-1} \textbf{v}$
\State $i \gets i + 1$
\EndWhile
\end{algorithmic}
\end{algorithm}
In this algorithm, we use the fact that in the fixed basis of $H_N$, the components of $\widehat M_N$ are given by
\begin{align}\label{eq:quantum_cat}
(\wh M)_{j,k} = \frac{1}{\sqrt N} \exp\left ( \pi i (2k ^2 - 2 j k + j^2 ) /N \right),
\end{align}
as computed in \cite{DJ}.
\end{prop}

Theorem \ref{t:torusquant} states that this distribution, $\mathbb P^{(n)} _{N,\rho}$, will converge (as $N \to \infty$) to $P_\mu ^{(n)}$. By Proposition \ref{prop:cs_beta_int}, the semi-classical defect measure $\mu$ for our chosen $\rho$ is $\delta _{x = x_0, \xi = p_0}$, so that
\begin{align*}
P_\mu^{(n)} (q_0 ,\dots ,q_n) = \prod _{ k = 0}^n \left | f \left (M^{k } \mat{x_0 \\ p_0} - q_k \right ) \right |^2 \ddd q_k.
\end{align*}
Therefore the distribution of $q_k$ according to $P_\mu ^{(n)}$ is a Gaussian centered at the position component of $M^k (x_0, p_ 0) ^t $, with variance $\sigma^2$.

Now we present several numerical simulations. 

First, Figure \ref{fig:fig2} demonstrates the evolution procedure. An initial state is given by \eqref{eq:initial state} with $N= 100$, $\beta = 0.5$, $\sigma' = 0.1$, $x_0 = 0.5$, and $p_0 = 0$. The symbols of the Kraus operators are given by $f(x -q)$ where $f$ is given by \eqref{eq:f_def}, with $\sigma = 0.1$. The absolute value squared of the components of $\psi$ (with respect to the basis of $H_N$) is plotted in blue. We then apply Algorithm \ref{qmap} to get a value $q_1$ which is plotted as an orange dot. The state is then changed by observation (by step 6 of Algorithm \ref{qmap}), which is plotted in the next plot. Then the particle is evolved by applying $\widehat M$, where $M$ is the cat map \eqref{eq:cat_map_matrix} to get a new state. The absolute value squared of the components of the new state are plotted in the next plot, and the process is repeated for 3 time steps.

\begin{figure}
 \centering
 \includegraphics[width = \textwidth]{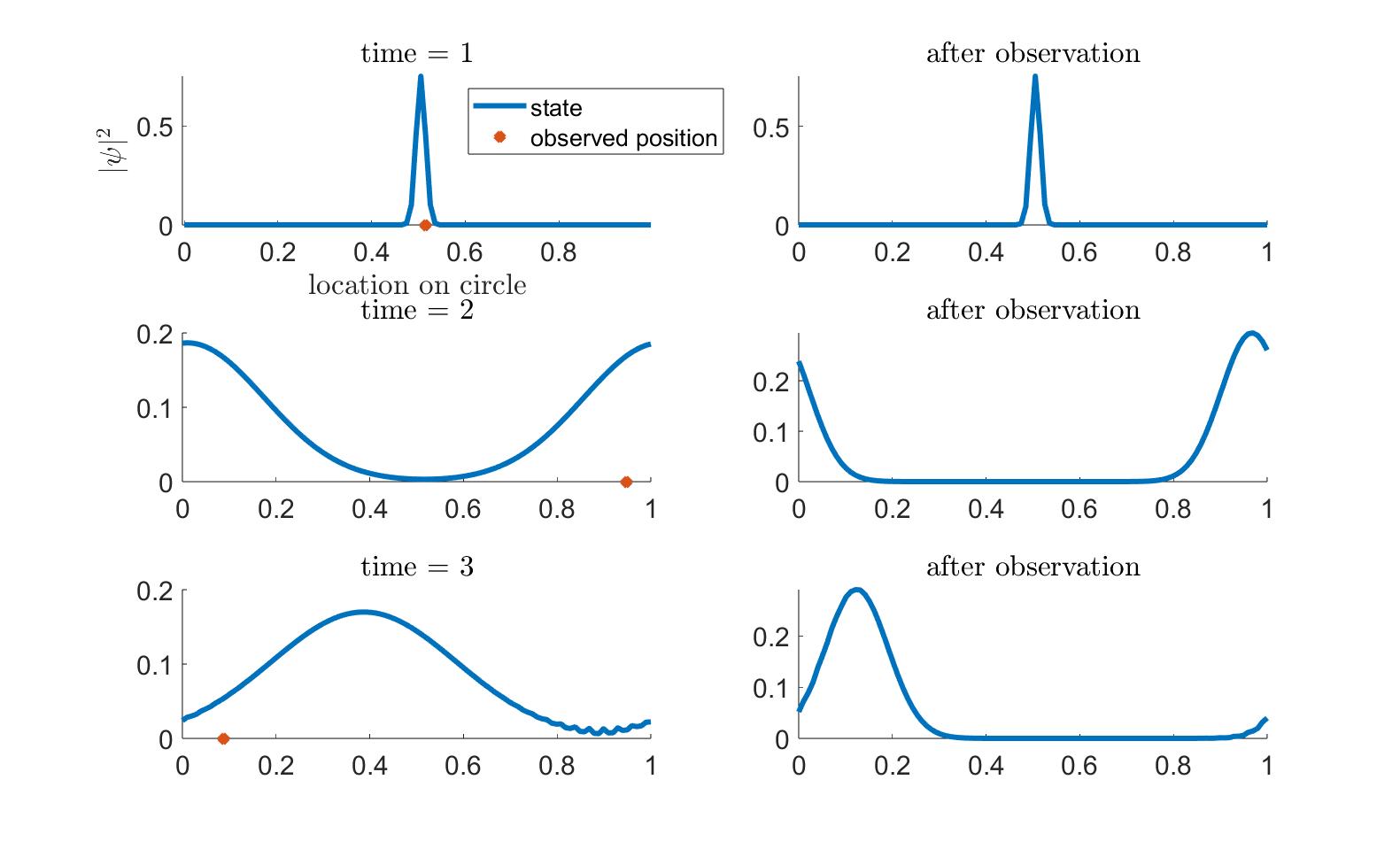}
 \caption{\textbf{Evolution Procedure} Here we simulate the trajectory of a quantum particle for three time steps under the discussed evolution procedure. The specifics of this simulation are discussed in \S \ref{s:numerics}.}
 \label{fig:fig2}
\end{figure}

Figure \ref{fig:fig3} compares an approximation of the marginal distributions of $\bP_{N, \rho}^{(n)} $ against $P_\mu^{(n)}$ for $n= 4$. For this simulation, the initial state is given by \eqref{eq:initial state} with $N= 2000$, $\beta = 0.5$, $\sigma' = 0.1$, $x_0 = 0.1$, and $p_0 = 0.1$. The symbols of the Kraus operators are given by $f(x -q)$ where $f$ is given by \eqref{eq:f_def}, with $\sigma = 0.1$. We then run the evolution procedure 2000 times (we call this the number of trials), saving the observed positions for each time step. The relative frequencies of the recorded positions are plotted as a histogram in blue. As the number of trials goes to infinity, these histograms should converge to the marginal distributions of $\mathbb P^{(n)} _{N, \rho}$ at each time step. For each instance in time, we plot the marginal distribution of $P_\mu^{(n)}$ as an orange curve. This paper's main result states that the two distributions should agree until the Ehrenfest time. For the cat map, the Lyapunov exponent is $\Gamma \approx .96$, so the Ehrenfest time is approximately $\log(N)/(2\Gamma) \approx 3.95$.

\begin{figure}
 \centering
 \includegraphics[width = \textwidth]{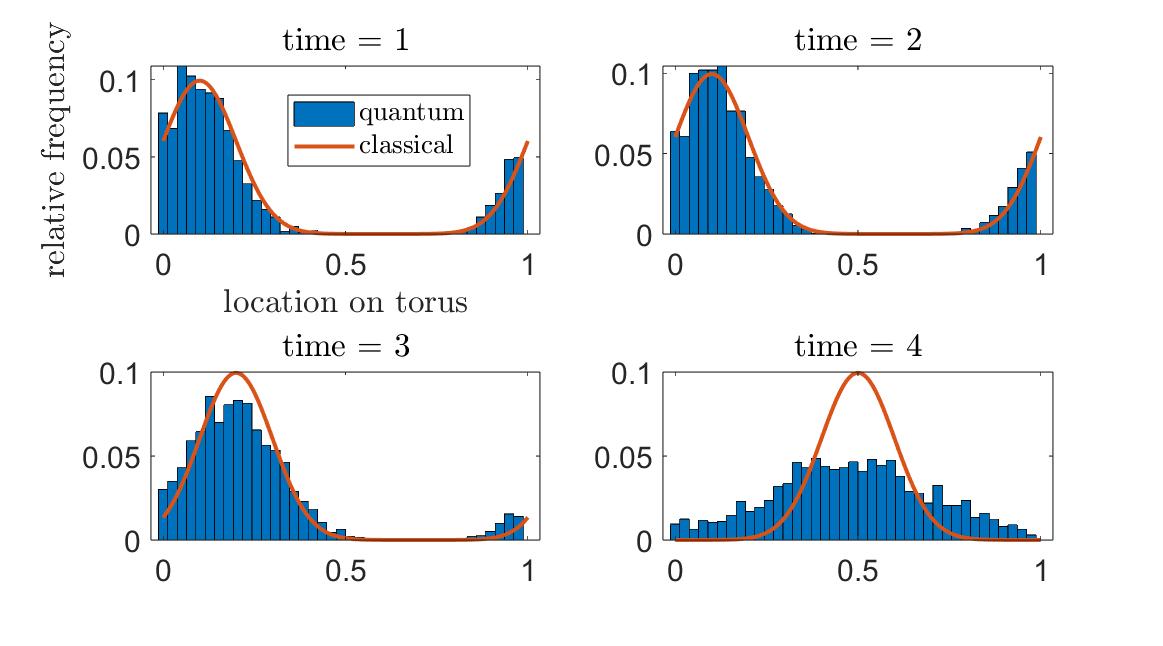}
 \caption{\textbf{Marginal distributions of $\bm{\bP_{N, \rho}^{(n)}} $ vs $\bm{P_\mu^{(n)}}$:} We perform Algorithm \ref{qmap} (observation then evolution by the quantum cat map). We plot the relative frequencies of the observed positions (approximating the marginal distribution of ${\bP_{N, \rho}^{(n)}}$ against the marginal distributions of ${P_\mu^{(n)}}$). See \S \ref{s:numerics} for more details.}
 \label{fig:fig3}
\end{figure}

It is apparent that the convergence of $\bP_{N, \rho}^{(n)}$ to $P_\mu^{(n)}$ as $N\to \infty$ is relatively slow. As seen in Figure \ref{fig:fig3}, after only $4$ time steps, the distribution of observed positions starts to become uniform. This should be expected, as Theorem \ref{t:torusquant} only guarantees \eqref{eq:thm_thing} for $N \gtrsim \exp((\Gamma + \varepsilon ) n / \delta ) $. For the cat map, $\Gamma = \log ((3 + \sqrt 5) /2) $, and in this case $\delta < 1/4$. Then, \eqref{eq:thm_thing} requires (roughly) $N \gtrsim 50^n$. These estimates are far from sharp, but give some sense of the exponentially slow rate of convergence as seen in Figure \ref{fig:fig3}.

In Figure \ref{fig:fig4}, we run the same simulation as in Figure \ref{fig:fig3} but vary $N$ from $100$ to $4000$. For each value of $N$, we compute the total variation of the (approximate) marginal distribution of $\bP_{N, \rho}^{(time)}$ against $P_\mu^{(time)}$ for times between $1$ and $6$. Observe that increasing $N$ leads to smaller total variation at an exponentially slow rate.

\begin{figure}
 \centering
 \includegraphics[width = \textwidth]{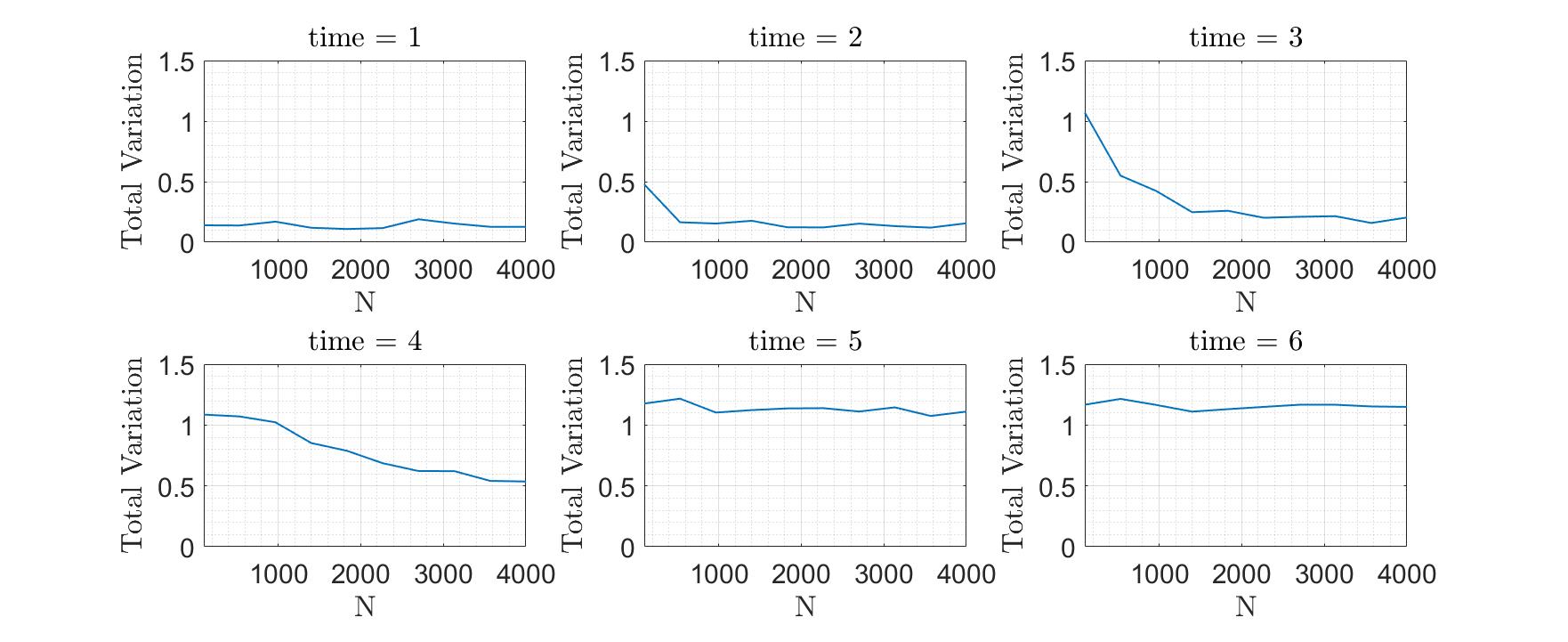}
 \caption{\textbf{Convergence of marginal distributions of $\bm{\bP_{N, \rho}^{(n)} \to P_\mu^{(n)}}$:} We run the same simulation as in Figure \ref{fig:fig3}, but vary $N$ from $100$ to $4000$ and compare $\bP_{N, \rho}^{(n)}$ to $ P_\mu^{(n)}$. See \S \ref{s:numerics} for more details.}
 \label{fig:fig4}
\end{figure}

Next, we show how changing the Lyapunov exponent changes the rate of convergence of $\bP_{N, \rho}^{(n)}$ to $P_\mu^{(n)}$. In Figure \ref{fig:fig5}, the same evolution procedure is performed for $4$ time steps but with $3$ different quantized symplectic matrices. In each case, we plot the relative frequencies of the observed positions (approximating the marginal distributions of $\mathbb P^{(n)}_{N,\rho}$) with a plot of the distribution of $P_\mu^{(n)}$. In all cases, we chose $N = 3000$, $\beta = 0.5$, $\sigma' = \sigma = 0.1$, $x_0 = 0.3$, $p_0 =0.9$, and simulated 500 trajectories. The symplectic matrices chosen were:
\begin{align}
 M_1 = \mat{2 & 1 \\ 1 & 1}, && M_2 = \mat{2 & 1 \\ 3 & 2}, && M_3 = \mat{4 & 1 \\ 15 & 4}, \label{eq:matrices}
\end{align}
which have Lyapunov exponent roughly $0.962, 1.317$, and $2.063$ respectively. Computing the Ehrenfest time as $\log(N) / (2\Gamma) $, then the three evolution procedures have Ehrenfest time approximately $4.16$, $3.04$, and $1.94$ respectively. This roughly agrees with the numerics provided. After only a few time steps, $\bP_{N, \rho}^{(n)}$ and $P_\mu^{(n)}$ disagree. Because the Ehrenfest time grows logarithmically with $N$, and the computations grow like $N^2$, it is not feasible to simulate numerics for which the two measures agree for many more time steps.

\begin{figure}
 \centering
 \includegraphics[width = \textwidth]{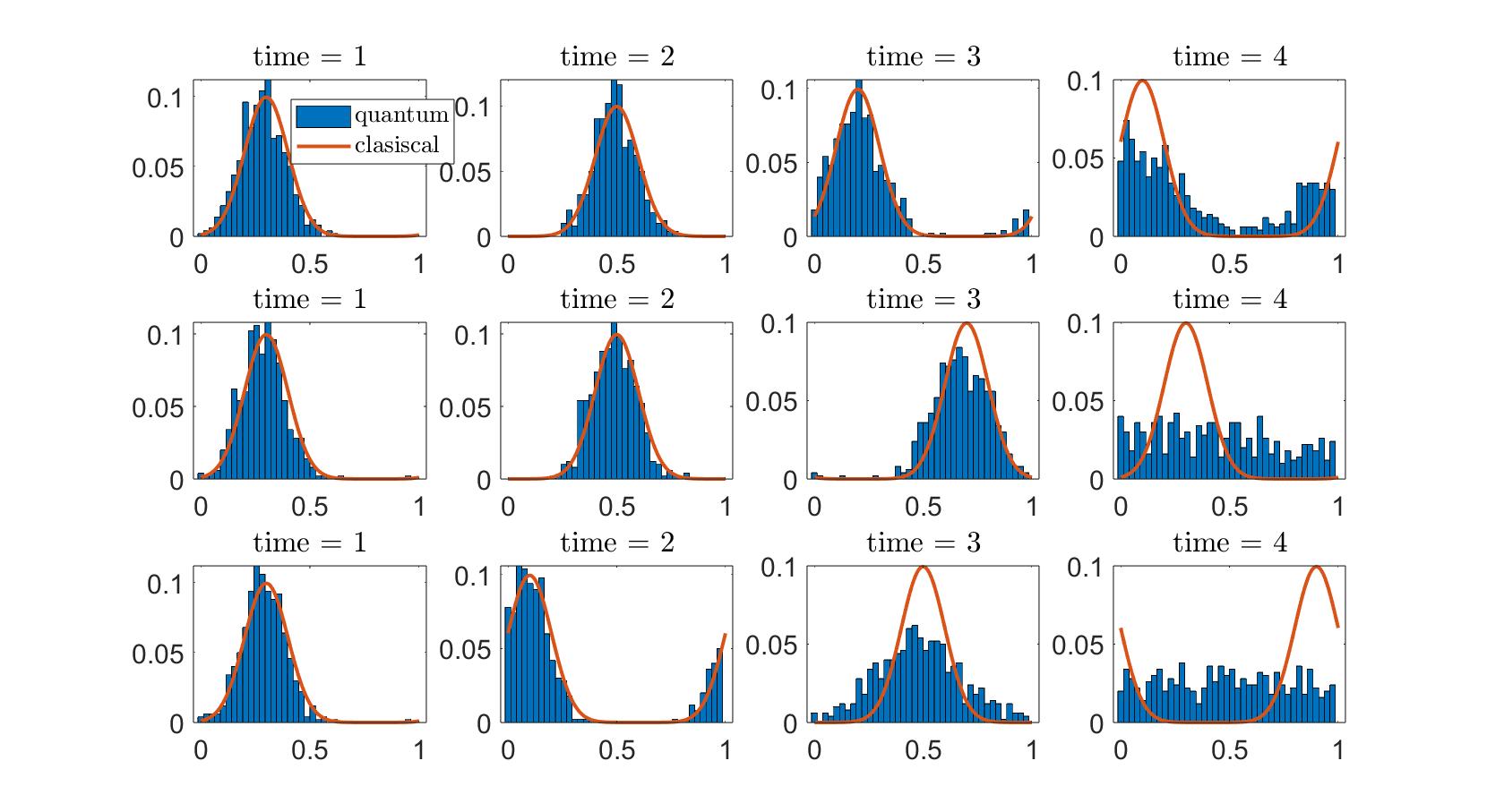}
 \caption{\textbf{Changing the Lyapunov exponent:} We present the same numerics as in Figure \ref{fig:fig3} but evolve by quantizations of three different matrices given by \eqref{eq:matrices} with 3 different Lyapunov exponents. See \S \ref{s:numerics} for more details.}
 \label{fig:fig5}
\end{figure}

In Figure \ref{fig:fig6} we replace the evolution of the particle by $\wh M_N$ at each time step, with evolution by $\exp( -2 \pi i t N \Op_N (a))$, where $a = \cos (2\pi x ) + \cos (2 \pi \xi)$, and $t=0.05$ -- which should be thought of as the interval of time between observations. In this case, $P_\mu ^{(n)}$ is a Gaussian whose mean is computed by evolving $(x_0,p_0)$ along the Hamiltonian flow generated by $a(x,\xi)$ for time $tn$. Note that the convergence of $\bP_{N, \rho}^{(n)}$ to $P_\mu^{(n)}$ is much more rapid in terms of $N$ (as long as $t$ is small). We chose the following parameters: $N= 1000$, $\beta = 0.5$, $\sigma'=\sigma = 0.1$, $x_0 = 0.4$, $p_0 = 0.7$, with $1000$ trials.

\begin{figure}
 \centering
 \includegraphics[width = \textwidth]{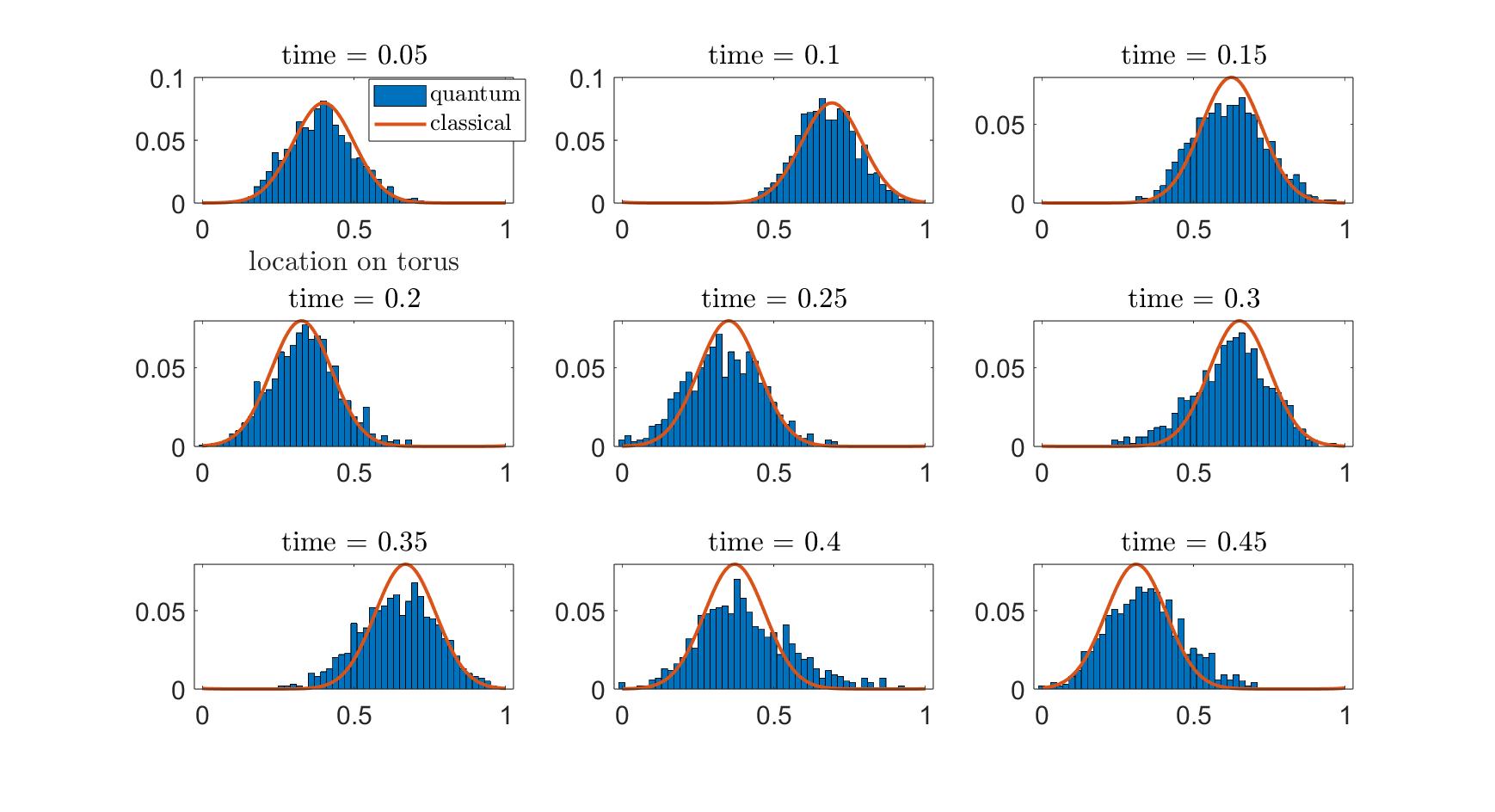}
 \caption{\textbf{Evolution of $\bm{\cos(2\pi x) + \cos (2\pi \xi)}$:} We present the same numerics as in Figure \ref{fig:fig3} but instead of evolving the quantum state by applying the quantum cat map, we multiply our state by $\exp(-2\pi it N \Op_N (\cos (2\pi x) + \cos (2\pi \xi) ))$. See \S \ref{s:numerics} for more details.}
 \label{fig:fig6}
\end{figure}

Lastly, in Figure \ref{fig:fig7}, we present numerics for the evolution of quantum particles, where the cat map is replaced by the matrix:
\begin{align}\label{eq:elliptic}
E:= \mat{-1 & 1 \\ -1 & 0}
\end{align}
which is an example of an elliptic matrix, having Lyapunov exponent equal to zero (in fact, $E^3 = I$). The parameters chosen here are $N= 1000$, $\beta = 0.5$, $\sigma = \sigma' = 0.1$, $x_0 =0.1$, $p_0 = 0.9$, with $1000$ trials. In this case $\bP_{N, \rho}^{(n)}$ and $P_\mu^{(n)}$ agree for more time steps than the evolution of the other symplectic maps. The Ehrenfest time will subtly depend on the $\varepsilon$ in \eqref{eq:ehrenfest_condition}.

\begin{figure}
 \centering
 \includegraphics[width = \textwidth]{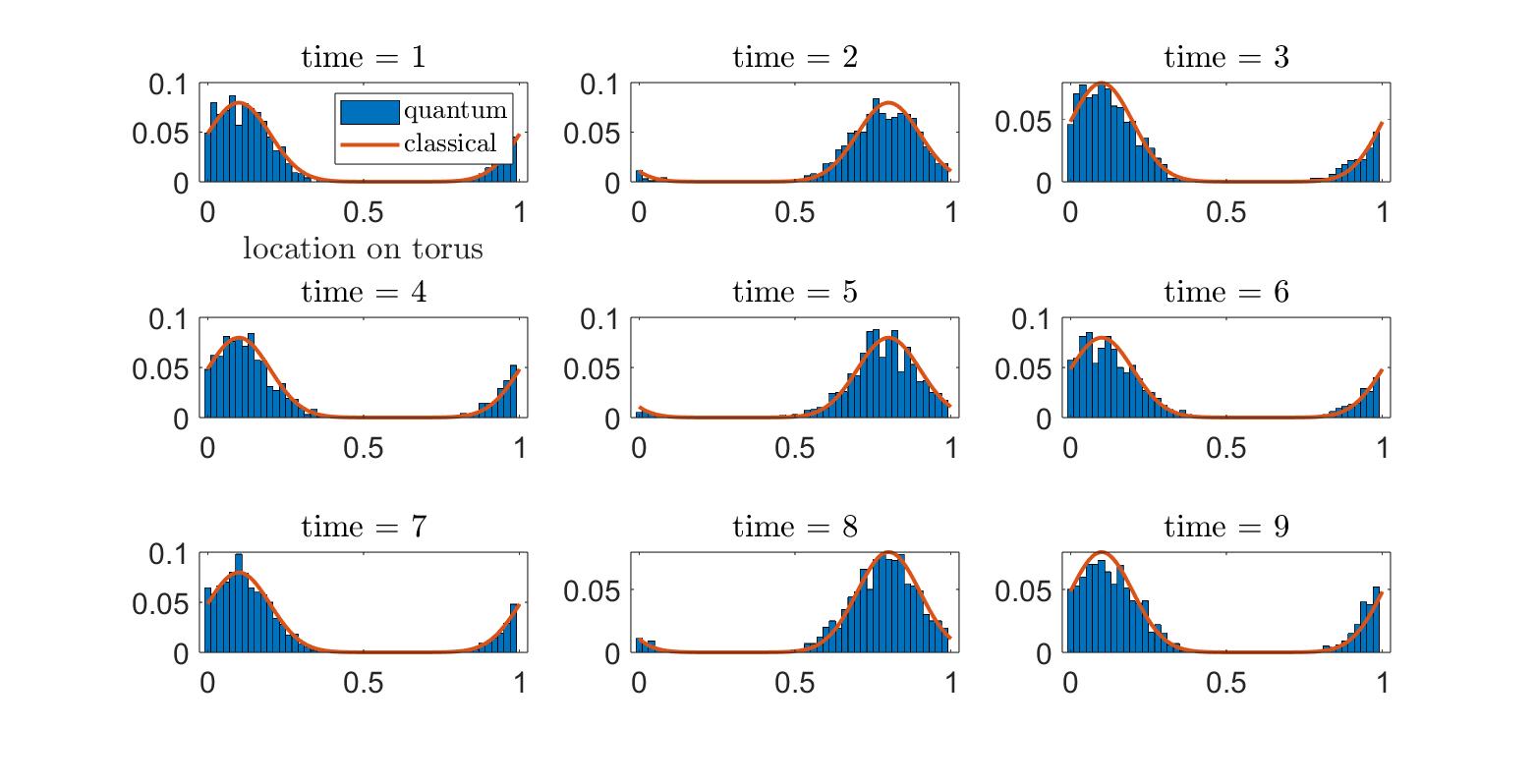}
 \caption{\textbf{Evolution by Elliptic Matrix:} We present the same numerics as in Figure \ref{fig:fig3} but instead of evolving by the quantum cat map, we evolve by the quantization of the elliptic matrix $E$ which has Lyapunov exponent zero. See \S \ref{s:numerics} for more details.}
 \label{fig:fig7}
\end{figure}

The figure on the first page (Figure \ref{fig:fig_first_page}) is simulated in the following way. We chose an initial state $\rho$ given in \eqref{eq:initial state} with $N = 500$, $\sigma' = 0.1$, $\beta = 0.1$, $x_0 = 0.4$, $p_0 = 0.7$. We indirectly measured the particle location using Kraus operators with symbol given by \eqref{eq:f_def} with $\sigma = 0.1$. Between each measurement, the particle is evolved by $\exp(-2\pi i t N P)$ with $P = \Op_N (\cos (2\pi  x) + \cos (2\pi \xi))$ and $t = 0.01$. This procedure is done for $100$ time steps to sample a single trajectory with law $\mathbb P^{(100)}_{N,\rho}$. We then repeat this $60$ times to get $60$ different trajectories, which are plotted on top of each other. We plot a single quantum trajectory in red. The corresponding classical probability measure is a product of Gaussian distributions with variance $\sigma^2$ and at time $t$ is centered at $\phi^t ((x_0,p_0))$. The flow $\phi^t$, the flow generated by the Hamiltonian vector field with Hamiltonian $\cos (2\pi x) + \cos (2\pi \xi)$, is numerically computed by solving 
\begin{align*}
    \begin{cases}
        \partial _t x(t) =  -2\pi \sin(2 \pi \xi(t)) \\
        \partial _t \xi(t) = 2\pi \sin (2\pi x(t))
    \end{cases}
\end{align*}
with initial conditions $x(0) = x_0$ and $\xi(0) = p_0$. We plot $x(t)$ as a blue dotted line.

\section{Acknowledgements}
The authors are grateful to Maciej Zworski and Martin Fraas for helpful discussions as well as two anonymous referees for numerous helpful corrections and suggestions. The authors also thank Semyon Dyatlov for providing a reference to proving Proposition \ref{p:Evolved_Symbol_Bound} and to Oliver Edtmair and Ian Gleason for fruitful discussions in writing up the proof in the appendix. This paper is based upon work supported by National Science Foundation grant DMS-1952939. The second author gratefully acknowledges support from the National Science Foundation Graduate Research Fellowship under grant DGE-1650114.

\appendix

\section{Uniformity of Defect Measures}\label{a:uniformDM}

Proposition \ref{p:DMExistence} motivates the definition of defect measures, but this definition allows for some undesirable consequences, which show up in the following pathological example. Let $f\in\cinf_0(\R^d)$ with $\|f\|_{L^2(\R^d)}=1$, and define $\psi_h(x)=f\left(x+ h^{-1}\right).$ It is easy to show that $\psi_h$ has defect measure equal to the zero measure, corresponding intuitively to the mass having ``escaped to infinity." Indeed, one has for any $a\in\cinf_0(T^*\R^d)$ supported in a ball of radius $R$ in $\R^{2d}\cong T^*\R^d$, $h<(2R)^{-1}$, and $\chi\in\cinf(\R^d)$ such that $\chi(x)=1$ on $|x|\ge2R$, $\chi=0$ on $|x|\le R$ that \begin{equation}\label{eq:badDM}\Op_h(a)\psi_h=\Op_h(a)\Op_h(\chi)\psi_h=\ohi
\end{equation}which gives
\begin{equation}\label{eq:badDM_limit}
 \lim_{h\to0}\la\Op_h(a)\psi_h,\psi_h\ra=0.
\end{equation}
Unfortunately, there is no way to quantify the rate of decay of (\ref{eq:badDM_limit}), as it depends on not just the symbol norms of $a$ but also the location of its support. In addition, the resulting measure is not a probability measure, which defies our intuition of the defect measure representing the classical location of the particle in phase space. 
We are therefore motivated to make a more restrictive definition of defect measures to prohibit examples such as (\ref{eq:badDM}). Call $\mu$ a \emph{uniform semiclassical defect measure} of $\rho_h$ if for any $a\in S$,
\begin{equation}\label{eq:UDF_Def}
 \lim_{h\to0}\tr\left(\rho_h\Op_h(a)\right)=\int_{T^*\R^d}a\ddd \mu.
\end{equation}
The reason for the name is given in the following proposition. We use the standard multi-index notation. For $a\in C^\infty(T^* \R^d)$ and $\alpha \in \N^{2d}$, we let 
\begin{align*}
    a^{(\alpha)} ( x) = \left(\prod _{j = 1}^{2d} \partial_j^{\alpha_j} \right)a(x).
\end{align*}
\begin{prop}\label{p:UDF_Equivalence}
Let $\rho_h$ be a set of density operators with defect measure $\mu$. The following are equivalent.
\begin{enumerate}
 \item $\mu$ is a uniform defect measure.
 \item For all $\epsilon>0$, there is an $h_0(\epsilon) > 0 $ such that if $0< h<h_0$ and $a\in\cinf_0(T^*\R^d)$ with $\|a^{(\alpha)}\|_{L^{\infty}}\le 1$ for all $|\alpha|<K(d)$ then $$\left|\tr\left(\rho_h\Op_h(a)\right)-\int_{T^*\R^d}a\ddd \mu\right|\le\epsilon.$$
 \item For all $\epsilon>0$, there is an $h_0(\epsilon)>0 $ such that if $0< h<h_0$ and $a\in S$ with $\|a^{(\alpha)}\|_{L^{\infty}}\le 1$ for all $|\alpha|<K(d)$ then $$\left|\tr\left(\rho_h\Op_h(a)\right)-\int_{T^*\R^d}a\ddd \mu\right|\le\epsilon.$$
 \item $\mu$ is a probability measure.
\end{enumerate}
\end{prop}
\begin{proof}
We trivially see that $(3)\implies(2)$ and $(3)\implies(1)$.

To show $(2)\implies(3)$, let $a\in S$, and let $\chi_k\in\cinf_0(T^*\R^d)$ be smooth partition of unity of $T^*\R^d$, chosen as in \cite[Theorem 4.23]{Zworski}, with the first $K(d)$ derivatives chosen small enough. By the Cotlar-Stein-Knapp lemma, we have the limit in the strong operator topology: $$\lim_{K\to\infty}\left(\sum_{k=0}^{K}\Op_h(\chi_ka)\right)\psi=\Op_h(a)\psi$$ for every $\psi\in L^2(\R^d)$. We need the following brief lemma.
\begin{lem}\label{l:tracelimit}
Let $A_K$ be a sequence of uniformly bounded operators converging to $A$ in the strong topology, and let $\rho$ be self-adjoint and trace class. Then $$\lim_{K\to\infty}\tr(\rho A_K)=\tr(\rho A).$$
\end{lem}
\begin{proof}[Proof of Lemma \ref{l:tracelimit}]
By subtracting we may assume $A=0$, and without loss of generality let $\sup_K\|A_K\|=1$. Let $\psi_j$ be an orthonormal basis of eigenvectors of the compact operator $\rho$, so $$\rho=\sum_j\lambda_j\ket{\psi_j}\bra{\psi_j}$$ with $\lambda_j\to0$. as $j\to\infty$. Let $\epsilon>0$, and choose $J$ such that $\lambda_j<\eps$ for $j>J$. Then
\begin{equation*}
 |\tr(\rho A_K)|=\left|\sum_{j\le J}\gl_j\la A_K\psi_j,\psi_j\ra+\sum_{j\le J}\gl_j\la A_K\psi_j,\psi_j\ra\right|
 \le \sum_{j\le J}\left|\gl_j\la A_K\psi_j,\psi_j\ra\right|+\eps
 \xrightarrow{K\to\infty}\eps
\end{equation*}
which shows $\tr(\rho A_K)\xrightarrow{K\to\infty}0$ as desired.
\end{proof}

By Lemma \ref{l:tracelimit}, we see that for fixed $h$, $$\lim_{K\to\infty}\tr\left(\rho_h\sum_{k=0}^{K}\Op_h(\chi_ka)\right)=\tr(\rho_h\Op_h(a))$$ and similarly as $\mu$ is a finite measure $$\lim_{K\to\infty}\int_{T^*\R^d}\sum_{k=0}^{K}\chi_ka\ddd \mu=\int_{T^*\R^d}a\ddd \mu.$$ Then letting $\eps>0$ and picking the $h_0(\eps)$ from part $(2)$ we get
\begin{equation*}
 \left|\tr\left(\rho_h\Op_h(a)\right)-\int_{T^*\R^d}a\ddd \mu\right|=\lim_{K\to\infty}\left|\tr\left(\rho_h\sum_{k=0}^{K}\Op_h(\chi_ka)\right)-\int_{T^*\R^d}\sum_{k=0}^{K}\chi_ka\ddd \mu\right|\le\epsilon
\end{equation*}

To show $(1)\implies(3)$, let $a\in S$, and let $\chi\in\cinf_0(T^*\R^d)$ be supported in a ball of radius $R$, such that $\mu(T^*\R^d\setminus B_R(0))\le \epsilon / 8$. We also note that $\chi$ may be chosen so its first $K(d)$ derivatives are small (possibly making $R$ larger). Then we have that
\begin{equation*}\label{eq:DM_farfrom0_bound}
 \tr(\Op_h(a)\Op_h(1-\chi)\rho_h)\le \| \Op_h (a) \| \tr(\Op_h(1-\chi)\rho_h)\le C \tr(\Op_h(1-\chi)\rho_h)
\end{equation*}
for $C$ depending on the first $K(d)$ derivatives of $a$. Next we apply (\ref{eq:UDF_Def}) with the symbol $1-\chi$ (which is independent of $a$) to get that
\begin{align*}
C \tr(\Op_h(1-\chi)\rho_h) = C \varepsilon / 8 + o (1)
\end{align*}
as $h \to 0$. So that there exists $h_0 >0$ such that for $0 < h < h_0$ we have $C \tr(\Op_h(1-\chi)\rho_h) < \varepsilon / 4$. With this, the triangle inequality, and symbol calculus, we have for $h<h_0(\eps)$
\begin{align}\label{eq:splitting_with_chi}
\begin{split}
\left|\tr\left(\rho_h\Op_h(a)\right)-\int_{T^*\R^d}a\ddd \mu\right|&\le\left|\tr\left(\rho_h\Op_h(a)\Op_h(\chi)\right)-\int_{T^*\R^d}a\chi\ddd \mu\right|+\frac{\eps}{2}\\
&\le\left|\tr\left(\rho_h\Op_h(a\chi)\right)-\int_{T^*\R^d}a\chi\ddd \mu\right|+\frac{3\eps}{4}
\end{split}
\end{align}
where for the last inequality we have chosen possibly smaller $h_0(\eps)$ and used the fact that $a$ has bounded derivatives.

We now prove $(3)$ using a compactness argument. Observe that by Calder\'on-Vaillancourt, if (\ref{eq:UDF_Def}) holds for $a\in S$, it must also hold for any $a\in C^{K_0(d)}(T^*\R^d)$ for some sufficiently large $K_0$. Given $h_j$, let $X_j$ be the set of $b\in C^{K_0(d)}(B_R(0))$ such that
\begin{equation}\label{eq:X1_bound}
 \left|\tr(\rho_h\Op_h(b\chi))-\int_{B_R(0)}b\chi\ddd \mu\right|<\frac{\eps}{4}
\end{equation}
for all $h<h_j$. Then for $K_0(d)$ large enough, condition (\ref{eq:X1_bound}) makes $X_j$ to be an open set in $C^{K_0(d)}(B_R(0))$, so letting $h_j\to0$, gives $(X_j)_{j\in\N}$ to be an open cover of $C^{K_0(d)}(B_R(0))$. Let $$Y=\{b\in\cinf(B_R(0)):\|b^{(\alpha)}\|_{L^{\infty}(B_R(0))}\le 1\quad\forall\,\,|\ga|\le K_0(d)+1\},$$ and let $\overline{Y}$ denote the closure of $Y$ in $C^{K_0(d)}(B_R(0))$. By the Arzela-Ascoli theorem, $\overline{Y}$ is compact in the topology of $C^{K_0(d)}(B_R(0))$, so there is a finite subset of the $X_j's$ that covers $\overline{Y}$, and hence an $h_0$ such that for $h<h_0$ $$\left|\tr(\rho_h\Op_h(b))-\int_{B_R(0)}b\ddd \mu\right|<\frac{\eps}{4}$$ holds for all $b$ with $\|b^{(\ga)}\|_{L^{\infty}}\le 1$ for $|\ga|\le K_0(d) + 1$. Letting $K(d)=K_0(d)+1$, and $b=a|_{B_R(0)}$ in (\ref{eq:splitting_with_chi}) completes the proof of $(3)$.

To show $(3)\implies(4)$, we simply see that $$\int_{T^*\R^d}\ddd \mu=\lim_{h\to0}\tr(\rho_h\Op_h(1))=\lim_{h\to0}\tr\rho_h=1.$$
Here we use that the function $f(x) = 1$ is in the symbol class $S$.

Finally, we show $(4)\implies (3)$. Let $\chi\in\cinf_0(T^*\R^d)$, so $$\tr\left(\rho_h\Op_h(1-\chi)\right)-\int_{T^*\R^d}1-\chi\ddd \mu=-\left(\tr\left(\rho_h\Op_h(\chi)\right)-\int_{T^*\R^d}\chi\ddd \mu\right)\to0$$ as $h\to0$, so in particular (\ref{eq:UDF_Def}) holds in the special case of the noncompact symbol $1-\chi$. But that was the only symbol used in the proof of $(1)\implies(3)$, so in particular shows $(4)\implies (3)$ as well. This shows all equivalences and hence completes the proof.
\end{proof}
\section{Coherent States}\label{a:coherent}
In this appendix, we prove Propositions \ref{p:Rdcoherent_strong}, \ref{prop_cs_beta0}, and \ref{prop:cs_beta_int}, which give examples of $(\gd,\gt)$-defect measures.
\subsection{Coherent States on $\R^d$}\label{aa:coherent_Rd}
We prove Proposition \ref{p:Rdcoherent_strong}. The proof is essentially the same as \cite[\S 5.1 Examples 1 and 2]{Zworski}, but is quantitative and keeps track of errors.

\begin{proof}[Proof of Proposition \ref{p:Rdcoherent_strong}]
We first consider the case when $\beta=0$. In this case, we can without loss of generality assume $x_0=0$, by absorbing it into the definition of $g$. Then $$\psi_h(x)=g(x)e^{\frac{i}{h}\la x,\xi_0\ra},$$
so for $a\in S_{\gd} \cap C_0^\infty (T^* \mathbb R^d)$:
\begin{align*}
 \tr(\rho_h\Op_h(a))&=\la\Op_h(a)\psi_h,\psi_h\ra\\
 &=\int\frac{1}{(2\pi h)^d}\left(\iint e^{\frac{i}{h}\la x-y,\xi-\xi_0\ra} g(y)\,a\left(\frac{x+y}{2},\xi\right)\ddd y\ddd \xi\right)\overline{g(x)}\ddd x\\
 &=\int\frac{1}{(2\pi h)^d}\left(\iint e^{\frac{i}{h}\la y,\xi\ra} g(x-y)\,a\left(x-\frac{y}{2},\xi+\xi_0\right)\ddd y\ddd \xi\right)\overline{g(x)}\ddd x.
\end{align*}
We evaluate the inner double integral using the explicit form of stationary phase given in \cite[Theorem 3.17]{Zworski}.
\begin{equation*}\begin{gathered}\iint e^{\frac{i}{h}\la y,\xi\ra} g(x-y)\,a\left(x-\frac{y}{2},\xi+\xi_0\right)\ddd y\ddd \xi\\\sim(2\pi h)^d\left(\sum_{k=0}^{\infty}\frac{h^k}{i^kk!}\la D_y,D_{\xi}\ra^k\left.g(x-\cdot_y)a\left(x-\frac{\cdot_y}{2},\cdot_{\xi}+\xi_0\right)\right|_{y=0,\xi=0}\right)\\=(2\pi h)^d\left(g(x)a(x,\xi_0)+O(h^{1-2\gd})\right)
\end{gathered}\end{equation*}
with $\sim$ referring to asymptotic summation and the last equality coming from $a$ being an element of $S_{\delta}$. Then
\begin{equation*}
 \left|\tr(\rho_h\Op_h(a))-\int|g(x)|^2a(x,\xi_0)\ddd x\right|\le O(h^{1-2\gd})\int|g(x)|\ddd x=O(h^{1-2\gd})
\end{equation*}
which gives the result for $\gb=0$. The case of $\gb=1$ follows from taking the semiclassical Fourier transform.

We now treat the case when $0<\gb<1$. Here we have $$\psi_h(x)=h^{-\frac{d\gb}{2}}g\left(\frac{x-x_0}{h^{\gb}}\right)e^{\frac{i}{h}\la x,\xi_0\ra}.$$ We again have that $\tr(\rho_h \Op_h (a) ) = \langle \Op_h(a) \psi _h , \psi _h \rangle $. We now approximate $\Op_h(a)$ by its left quantization defined as:
$$\Op_h^L(a)u(x):=\frac{1}{(2\pi h)^d}\iint e^{\frac{i}{h}\la x-y,\xi\ra}a(x,\xi)u(y)\ddd y\ddd \xi.$$ By \cite[Theorem 4.13]{Zworski}, $\Op_h(a)=\Op_h^L(a)+O(h^{1-2\gd})_{L^2(\R^d)\to L^2(\R^d)}$. Therefore
\begin{align*}
 &\tr(\rho_h\Op_h(a))=\la\Op_h^L(a)\psi_h,\psi_h\ra+O(h^{1-2\gd})\\
 &=\frac{1}{(2\pi h)^d}\int\overline{g\left(\frac{x-x_0}{h^{\gb}}\right)}\left(\iint h^{-d\gb}g\left(\frac{y-x_0}{h^{\gb}}\right)e^{\frac{i}{h}\la x-y,\xi-\xi_0\ra}a(x,\xi)\ddd y\ddd \xi\right)\ddd x\\
 &\quad+O(h^{1-2\gd})\\
 &=\frac{1}{(2\pi h)^d}\int\overline{g\left(\frac{x-x_0}{h^{\gb}}\right)}\int e^{\frac{i}{h}\la x-x_0,\xi\ra}a(x,\xi+\xi_0)\left(\int h^{-d\gb}g\left(\frac{y}{h^{\gb}}\right)e^{-\frac{i}{h}\la y,\xi\ra}\ddd y\right)\ddd \xi\ddd x\\
 &\quad+O(h^{1-2\gd})\\
 &=\frac{1}{(2\pi h)^d}\iint\overline{g\left(\frac{x}{h^{\gb}}\right)}e^{\frac{i}{h}\la x,\xi\ra}a(x+x_0,\xi+\xi_0)\hat{g}\left(\frac{\xi}{h^{1-\gb}}\right)\ddd \xi\ddd x+O(h^{1-2\gd})\\
 &= \frac{1}{(2\pi)^d} \iint\overline{g(x)}\hat{g}(\xi)e^{i\la x,\xi\ra}a\left(h^{\gb}x+x_0,h^{1-\gb}\xi+\xi_0\right)\ddd \xi\ddd x+O(h^{1-2\gd}).
\end{align*}

Because $\| g \| _{L^2 (\mathbb R^d )} = 1$, we have by the Fourier inversion formula that:
\begin{align*}
    \frac{1}{(2\pi)^d} \iint \overline{g(x)}\hat{g}(\xi) e^{i \langle x , \xi \rangle} \ddd \xi \ddd x = 1
\end{align*}
so that
\begin{align*}
 &|\tr(\rho_h\Op_h(a))-a(x_0,\xi_0)| \\
 &\le \left| \frac{1}{(2\pi)^d} \iint\overline{g(x)}\hat{g}(\xi)e^{i\la x,\xi\ra}(a\left(h^{\gb}x+x_0,h^{1-\gb}\xi+\xi_0\right) - a(x_0,\xi_0) )\ddd \xi\ddd x  \right| + O(h^{1-2\delta}).
\end{align*}
By the mean value theorem and the fact that $a\in S_\delta \cap C_0^\infty (T^* \mathbb R^d)$, we have for $x,\xi \in \mathbb R ^d$,
\begin{align*}
    |a\left(h^{\gb}x+x_0,h^{1-\gb}\xi+\xi_0\right) - a(x_0,\xi_0) | &\le  (h^{\gb} | x | + h^{1-\gb} | \xi |) \| \nabla a \| _{L^\infty}  \\
&\le C  (h^{\gb} | x | + h^{1-\gb} | \xi | ) h^{-\delta}.
\end{align*}
So, using that $g$ is Schwartz,
\begin{align*}
|\tr(\rho_h\Op_h(a))-a(x_0,\xi_0)| &\le C \iint |g(x)||\hat g(\xi) |x||\xi| ((h^{\gb} | x | + h^{1-\gb} | \xi | ) h^{-\delta})\ddd \xi \ddd x\\
&= O(h^{\gb-\gd})+O(h^{1-\gb-\gd})+O(h^{1-2\gd}).
\end{align*}
Noting that $1-2\gd=\gb-\gd+1-\gb-\gd$ shows the $O(h^{1-2\gd})$ term to be superfluous and gives the desired result.
\end{proof}

\subsection{Coherent States on the Quantized Torus}\label{aa:coherent_torus}

We provide full proofs of Propositions \ref{prop_cs_beta0} and \ref{prop:cs_beta_int}, though we do not claim our parameters to be optimal. 

\begin{proof}[Proof of Proposition \ref{prop_cs_beta0}]
For this proof, we work in the orthonormal basis $(Q_n)$. Then for $a\in S_{\gd}(\T^{2d})$, $\Op_N(a)$ has a representation as a matrix $A$ with 
\begin{equation}\label{e:AmjExpression}
A_{mj}=\sum_{l,r\in\Z^d}\hat{a}(l,j-m-rN)(-1)^{\la r,l\ra}e^{\pi i\frac{\la j+m,l\ra}{N}}.
\end{equation}
Let $M=\lfloor N^{\g}\rfloor$ for $\g\in(\gd,1)$. We see that if $|j-m|>M$ (in distance $\mod N$), then for any $k\in\N _{>d
}$
\begin{align}\label{eq:mj_far}
\begin{split}
 |A_{mj}|&\le\sum_{l,r\in\Z^d}|\hat{a}(l,j-m-rN)|\\
 &\le C\sum_{l,r\in\Z^d}\la l \ra^{-(d+1)}\|(\dd_x^{d+1}+1)\dd_{\xi}^ka\|_{L^{\infty}}| j-m-rN|^{-k} \\
 &\le C N^{(d+1+k)\delta} \sum_{r \in \mathbb Z^d} |j - m - rN|^{-k}.
 \end{split}
 \end{align}
Note that there exists a large constant $M_0 > 0 $ (independent of $N$) such that
\begin{align*}
    \sum_{\substack{r \in \mathbb Z^d\\ |r| > M_0}} |j - m - rN|^{-k}  \le 1.
\end{align*}
Therefore:
\begin{align*}
     \sum_{r \in \mathbb Z^d} |j - m - rN|^{-k} \le M_0 M^{-k} + 1 \le C  N^{-\gamma k }
\end{align*}
for sufficiently large $N$. Therefore  
 \begin{align*}
 \begin{split}
 |A_{mj}|&\le C N^{(d+1)\gd}N^{k\gd}N^{-k\g}=CN^{-k(\g-\gd)+(d+1)\gd}
 \end{split}
\end{align*}
where $C=C(a,\gd,k)$ is independent of $N$. Taking $k$ arbitrarily large gives that $|A_{mj}|=O(N^{-\infty})$, so (as with pseudodifferential operators on $\R^d$), the mass of the kernel is concentrated on the diagonal and falls off super-polynomially away from it. We also note that even when $j$ and $m$ are close, we may remove one of the sums in (\ref{e:AmjExpression}) up to an $O(N^{-\infty})$ error. Specifically we see that when $|j-m|\le M$, $k\in\N$
\begin{equation*}
 \left|A_{mj}-\sum_{l\in\Z^d}\hat{a}(l,j-m)e^{\pi i\frac{\la j+m,l\ra}{N}}\right|=O(N^{-\infty}).
\end{equation*}
Indeed
\begin{align*}
   \left|  A_{mj}-\sum_{l\in\Z^d}\hat{a}(l,j-m)e^{\pi i\frac{\la j+m,l\ra}{N}}\right | &=\left | \sum_{\substack{l,r\in\Z^d \\ r \neq 0}}\hat{a}(l,j-m-rN)(-1)^{\la r,l\ra}e^{\pi i\frac{\la j+m,l\ra}{N}} \right| \\
    &\le\sum_{\substack{l,r\in\Z^d \\ r \neq 0}} |\hat{a}(l,j-m-rN)|.
\end{align*}
Note that there exists $c_d > 0 $ such that for all $r \neq 0$,
\begin{align*}
    |j - m -rN| \ge c_d N^\gamma |r|
\end{align*}
therefore  for each $k\in \Z_{> d}$, using that $a \in S_\delta(\mathbb T^{2d})$,
\begin{align*}
    \sum_{\substack{l,r\in\Z^d \\ r \neq 0}} |\hat{a}(l,j-m-rN)| &\le C \sum_{\substack{l,r\in\Z^d \\ r \neq 0}} \la l \ra ^{-(d+1)}|j - m -rN|^{-k} \|(\dd_x^{d+1}+1)\dd_{\xi}^ka\|_{L^{\infty}} \\
    &\le C N^{\delta(d+1 +k)} \sum_{r\in \Z^d_{\neq 0}}c_d N^{-k\gamma} \la r \ra ^{-k}
\end{align*}

by the same reasoning as in (\ref{eq:mj_far}). By Fourier inversion in the first variable, we have in that case
\begin{equation}\label{eq:mj_close}
A_{mj}=\sF_2a\left(\frac{j+m}{2N},j-m\right)+O(N^{-\infty})
\end{equation} where $\sF_2$ is the Fourier coefficient in the second component (so $\sF_2a$ acts on $\T^d\times\Z^d$).

Then we compute 

\begin{align}\label{eq:trace_sum}
\begin{split}
 &\tr_{H_N}(\rho_N\Op_N(a))=C_N\bra{\psi(N)}\Op_N(a)\ket{\psi(N)}\\
 &=C_N\sum_{m,j\in(\Z/N\Z)^d}\frac{1}{N^d}\overline{g(m/N)}A_{mj}g(j/N)e^{2\pi i(j-m)\xi_0}\\
 &=C_N\sum_{|m-j|\le M}\frac{1}{N^d}\overline{g(m/N)}g(j/N)\int_{\T^d}a\left(\frac{j+m}{2N},\xi\right)e^{2\pi i(m-j)(\xi-\xi_0)}\ddd \xi +O(N^{-\infty})\\
 &=C_N\sum_{|m-j|\le M}\frac{1}{N^d}\overline{g(m/N)}g(j/N)\int_{\T^d}a\left(\frac{j+m}{2N},\xi+\xi_0\right)e^{2\pi i(m-j)\xi}\ddd \xi +O(N^{-\infty})
 \end{split}
\end{align}
But when $|j-m|\le M$, we have
$$\left|a\left(\frac{j+m}{2N},\xi+\xi_0\right)-a\left(\frac{m}{N},\xi+\xi_0\right)\right|\le MN^{-1}\|\dd a\|_{L^{\infty}}\le CMN^{-(1-\gd)}$$ and $$\left|g\left(\frac{m}{N}\right)-g\left(\frac{j}{N}\right)\right|\le CMN^{-1},$$ so \begin{equation}\begin{gathered}\label{eq:term_error}\overline{g(m/N)}g(j/N)\int_{\T^d}a\left(\frac{j+m}{2N},\xi+\xi_0\right)e^{2\pi i(m-j)\xi}\ddd \xi\\=|g(m/N)|^2\int_{\T^d}a\left(\frac{m}{N},\xi+\xi_0\right)e^{2\pi i(m-j)\xi}\ddd \xi+O(MN^{-(1-\gd)})\end{gathered}\end{equation}
As there are $O(N^dM^d)$ pairs $(j,m)$ with $|m-j|\le M$, the errors (\ref{eq:term_error}) sum in (\ref{eq:trace_sum}) to give
\begin{align}\label{eq:trace_split_sum}
\begin{split}
 &\tr_{H_N}(\rho_N\Op_N(a))\\&=\sum_{m\in(\Z/N\Z)^d}\frac{1}{N^d}|g(m/N)|^2\sum_{j:|m-j|\le M}\int_{\T^d}a\left(\frac{m}{N},\xi+\xi_0\right)e^{2\pi i(m-j)\xi}\ddd \xi +O(M^{d+1}N^{-(1-\gd)})
 \end{split}
\end{align}
where we have used Remark \ref{r:normalization} to remove the $C_N$.
The inner sum in (\ref{eq:trace_split_sum}) is just a partial Fourier series, and as $\g>\gd$, it differs from the full Fourier series by order $O(N^{-\infty})$. This gives
\begin{equation}\label{eq:trace_simple_sum}
 \tr_{H_N}(\rho_N\Op_N(a))=\sum_{m\in(\Z/N\Z)^d}\frac{1}{N^d}|g(m/N)|^2a\left(\frac{m}{N},\xi_0\right)+O(M^{d+1}N^{-(1-\gd)}).
\end{equation}
But (\ref{eq:trace_simple_sum}) is just a Riemann sum, and we have the estimate $$\left|\sum_{m\in(\Z/N\Z)^d}\frac{1}{N^d}|g(m/N)|^2a\left(\frac{m}{N},\xi_0\right)-\int_{\T^d}a(x,\xi_0)|g(x)|^2\ddd x\right|=O(N^{-(1-\gd)}).$$
Therefore
\begin{equation*}
 \left|\tr_{H_N}(\rho_N\Op_N(a))-\int_{\T^d}a(x,\xi_0)|g(x)|^2\ddd x\right|\le O(M^{d+1}N^{-(1-\gd)})=(N^{(d+1)\g-(1-\gd)})
\end{equation*}
for all $\g>\gd$.
\end{proof}

\begin{proof}[Proof of Proposition \ref{prop:cs_beta_int}]
Let $a\in S_{\gd}(\T^{2d})$ and let $A=\Op_N(a)$ given in the basis $(Q_n)$. Once again choose $M=\lfloor N^{\g}\rfloor$ for $\g\in(\gd,1)$. We have by (\ref{eq:mj_far}) and (\ref{eq:mj_close}) that (writing $\psi:=\psi_N$ for convenience and letting subscripts refer to the indices)
\begin{align}\label{eq:trace_beta_first_sum}
\begin{split}
 \tr(\rho_N\Op_N(a))&=C_N'\sum_{m,j\in(\Z/N\Z)^d}\overline{\psi_m}A_{mj}\psi_j\\
 &=C_N'\sum_{|m-j|\le M}\overline{\psi_m}\int_{\T^d}a\left(\frac{j+m}{2N},\xi\right)e^{2\pi i(m-j)\xi}\ddd \xi\psi_j+O(N^{-\infty}).
\end{split}
\end{align}
But as $g\in\cS(\R^d)$, it has super-polynomial decay, and we get
\begin{align}\label{eq:vector_app}
\begin{split}
\psi_m&=N^{-\frac{d}{2}}\sum_{k\in\Z^d}g_N\left(\frac{m}{N}+k\right)\\
&=N^{\frac{d(\gb-1)}{2}}\sum_{k\in\Z^d}g\left(N^{\gb}\left(\frac{m}{N}+k-x_0\right)\right)e^{2\pi i(m+kN)\cdot\xi_0}\\
&=N^{\frac{d(\gb-1)}{2}}g\left(N^{\gb}\left(\frac{m}{N}-x_0\right)\right)e^{2\pi im\cdot\xi_0}+O(N^{-\infty})
\end{split}
\end{align} by throwing away all but the first term in the sum.
In fact, for many $m$ we may even throw out the first term. Given any $\epsilon>0$, let $R_{\eps}=\lfloor N^{1-\gb+\eps}\rfloor$. Then if $|m-Nx_0|\le R_{\eps}$, we have by (\ref{eq:vector_app}) that
\begin{equation}\label{eq:vector_tail}
 |\psi_m|=O(N^{-\infty}).
\end{equation}
We then approximate terms in the summand as in Proposition \ref{prop_cs_beta0}, and get when $|j-m|\le M$ that $$\left|a\left(\frac{j+m}{2N},\xi+\xi_0\right)-a\left(\frac{m}{N},\xi+\xi_0\right)\right|\le MN^{-1}\|\dd a\|_{L^{\infty}}\le CMN^{-(1-\gd)}$$ and $$\left|g\left(N^{\gb}\left(\frac{m}{N}-x_0\right)\right)-g\left(N^{\gb}\left(\frac{j}{N}-x_0\right)\right)\right|\le CMN^{-(1-\gb)}$$
so in that case (using $\gd\le \gb$)
\begin{equation}\label{eq:beta_summand_error}
 \begin{gathered}
 \overline{g\left(N^{\gb}\left(\frac{m}{N}-x_0\right)\right)}g\left(N^{\gb}\left(\frac{j}{N}-x_0\right)\right)\int_{\T^d}a\left(\frac{j+m}{2N},\xi+\xi_0\right)e^{2\pi i(m-j)\xi}\ddd \xi\\=\left|g\left(N^{\gb}\left(\frac{m}{N}-x_0\right)\right)\right|^2\int_{\T^d}a\left(\frac{m}{N},\xi-\xi_0\right)e^{2\pi i(j-m)\xi}\ddd \xi+O(MN^{1-\gb})
 \end{gathered}
\end{equation}
But by (\ref{eq:trace_beta_first_sum}), (\ref{eq:vector_app}), and (\ref{eq:vector_tail}) we have $\tr(\rho_N\Op_N(a))$ to be equal to
\begin{equation}\label{eq:trace_beta_full_sum}
\begin{gathered}
 C_N'N^{d(\gb-1)}\sum_{m:\left|m-Nx_0\right|\le R_{\eps}}\sum_{j:|m-j|\le M}\overline{g\left(N^{\gb}\left(\frac{m}{N}-x_0\right)\right)}g\left(N^{\gb}\left(\frac{j}{N}-x_0\right)\right)\cdots\\\cdots\int_{\T^d}a\left(\frac{j+m}{2N},\xi\right)e^{2\pi i(m-j)(\xi-\xi_0)}\ddd \xi+O(N^{-\infty})
 \end{gathered}
 \end{equation}
The sum in (\ref{eq:trace_beta_full_sum}) has $O(M^dR_{\eps}^d)$ terms, so adding up all the errors given in (\ref{eq:beta_summand_error}) gives
\begin{align*}
 &\tr(\rho_N\Op_N(a))\\
 &=N^{d(\gb-1)}\sum_{m:\left|m-Nx_0\right|\le R_{\eps}}\left|g\left(N^{\gb}\left(\frac{m}{N}-x_0\right)\right)\right|^2\sum_{j:|m-j|\le M}\int_{\T^d}a\left(\frac{m}{N},\xi+\xi_0\right)e^{2\pi i(m-j)\xi}\ddd \xi\\&\quad+O(M^{d+1}N^{(d+1)(\gb-1)}R_{\eps}^d)\\
 &=N^{d(\gb-1)}\sum_{m:\left|m-Nx_0\right|\le R_{\eps}}\left|g\left(N^{\gb}\left(\frac{m}{N}-x_0\right)\right)\right|^2a\left(\frac{m}{N},\xi_0\right)+O(M^{d+1}N^{\gb-1}N^{d\eps})
\end{align*}
where in the last line we approximated the inner sum with a Fourier series as in (\ref{eq:trace_simple_sum}), obtaining an $O(N^{-\infty})$ error, and applied Remark \ref{r:normalization} to ignore the $C_N'$. At this point, we (up to $O(N^{-\infty})$ error) add back in the remaining $m$ terms and compute the Riemann sum to get
\begin{align*}
 &\tr(\rho_N\Op_N(a))\\
 &=N^{d(\gb-1)}\sum_{m\in(\Z/N\Z)^d}\left|g\left(N^{\gb}\left(\frac{m}{N}-x_0\right)\right)\right|^2a\left(\frac{m}{N},\xi_0\right)+O(M^{d+1}N^{\gb-1}N^{d\eps})\\
 &=\int N^{d\gb}|g(N^{\gb}(x-x_0))|^2a(x,\xi_0)\ddd x+O(N^{\gb-1})+O(N^{\gd-1})+O(M^{d+1}N^{\gb-1}N^{d\eps})\\
 &=\int N^{d\gb}|g(N^{\gb}(x-x_0))|^2a(x,\xi_0)\ddd x+O(M^{d+1}N^{\gb-1}N^{d\eps})
\end{align*}
with the last line using that $\gd\le\gb$. We substitute and use $\int|g|^2=1$ to obtain
\begin{equation*}
\begin{gathered}
 \left|\int_{\T^d} N^{d\gb}|g(N^{\gb}(x-x_0))|^2a(x,\xi_0)\ddd x-a(x_0,\xi_0)\right|\\\le\left|\int_{\R^d}|g(y)|^2\left(a(N^{-\gb}y+x_0,\xi_0)-a(x_0,\xi_0)\right)dy\right|+\left|\int_{\R^d\setminus N^{\gb}\T^d}|g(y)|^2\ddd y\right|\\
 \le C\left|\int_{\R^d}|g(y)|^2N^{-\gb}N^{\gd}|y|\ddd y\right|+O(N^{-\infty})\le O(N^{\gd-\gb}).
\end{gathered}
\end{equation*}
This gives $$|\tr(\rho_N\Op_N(a))-a(x_0,\xi_0)|\le O(N^{\gd-\gb})+O(M^{d+1}N^{\gb-1}N^{d\eps})$$ which completes the proof by taking $\eps$ and $\g-\gd$ small enough.
\end{proof}

\section{Estimating Derivatives of a Flow}
\label{s:derivative of flow}

Here we present a proof of Proposition \ref{p:Evolved_Symbol_Bound} which we restate here for convenience. 

\begin{prop}\label{prop:931}Suppose $\phi^t : \mathbb T^{2d} \to \mathbb T^{2d}$ is a smooth dynamical system on $\mathbb T^{2d}$ with Lyapunov exponent $\Lambda$ where $t$ either takes integer values, or values in $\mathbb R_{\ge 0}$. Let $a\in C^\infty (\mathbb{T}^{2d})$. Then for all $\varepsilon >0$ and $\alpha \in \mathbb N^{2d}$, there exists $C = C(\varepsilon,\alpha)$ such that:
\begin{align*}
    |\dd^\alpha (a\circ \phi^t ) (x) | \le C e^{(\Lambda+\varepsilon) t |\alpha|} \| a \| _{C^{|\alpha|}} .
\end{align*}
\end{prop}

Here $C^k$ is the norm defined as
\begin{align*}
    \| a \|_{C^K } = \sum_{|\alpha |\le k} \| \dd^\alpha a\|_{L^\infty(\T^{2d})}.
\end{align*}

This proof of Proposition \ref{prop:931} essentially follows \cite[\S 5.2]{dyatlov2022} but in the simpler setting of the torus. To our knowledge, such a proof in this setting is not present in the current literature. In an effort to improve readability, we include examples throughout the proof, at the cost of being rather voluminous.

\begin{proof}

We prove Proposition \ref{prop:931} by strong induction on the order of $\alpha$. We trivially have the claim when $|\alpha | = 0$. If $|\alpha | = 1$, let $i$ be such that $\partial ^\alpha = \partial_i$. Then we have:
    \begin{align*}
        \partial^\alpha _x  ( a \circ \phi^t)(x)  = \sum_{j = 1}^{2d} (\partial_j a)\circ(\phi^t (x))\partial _i (\phi_t(x))_j
    \end{align*}
    where $(\phi^t(x))_j$ is the $j$th component of $\phi^t(x)$. By the definition of the Lyapunov exponent, for any $\e > 0$, there exists $C = C(\e)> 0$ such that:
    \begin{align*}
        |\partial_i \phi^t (x) | \le C e^{(\Lambda + \e)t},
    \end{align*}
    so that
    \begin{align*}
        |\partial^\alpha _x  ( a \circ \phi^t)(x)| \le C e^{(\Lambda + \e)t}  \| a \|_{C^1}.
    \end{align*}
To ease notation in the remained of the proof, define $\lambda := \Lambda + \varepsilon$.

Before proving the inductive step, we show how to go from $\alpha$ of order $1$ to order $2$. We will use the following notation for Jacobians and tensors. If $f  : \R \to \R^{2d}$, then $Df$ will be the vector with $i$th component $\dd_i f$. Note that $Df$ will be a $(0,1)$ tensor. We define $D^2f$ as the $(0,2)$ tensor with $i,j$ entry $(D^2 f )_{i,j} = \dd_i \dd_j f$. If $\phi : \R^{2d} \to \R^{2d}$, then we let $D\phi$ be the $(1,1)$ tensor with $i,j$ entry $(D\phi)_{j}^i = \dd_j \phi_i $. We let $D^2 \phi$ denote the $(1,2)$ tensor with $i,j,k$ entry $(D^2\phi )_{j,k}^i = (\dd_j \dd_k \phi_i)$. 

Observe that:
\begin{align*}
(D(a \circ \phi^n) (x))_i=  \dd_i (a\circ \phi^n)(x) = \sum_j \dd_j ((a\circ \phi^{n-1}) \circ \phi^1(x)  ) (\dd_i \phi_j^1)(x)
\end{align*}
so that
\begin{align*}
D(a \circ \phi^n)(x) =( D(a\circ \phi^{n-1}) \circ \phi^1(x) ) D\phi^1(x) 
\end{align*}
where the product of two tensors is the contraction of the upper and lower components. That is, $( D(a\circ \phi^{n-1}) \circ \phi^1 )$ is a $(0,1)$ tensor which can be written $A_{\mu}$, and $D\phi^1 $ is a $(1,1)$ tensor which can be written $B_{\gamma}^\nu$. In this case $AB$ is the $(0,1)$ tensor contracting the lower component of $A$ with the upper component of $B$: $(AB)_\gamma = \sum_\nu A_\nu B^\nu_\gamma$.

Similarly
\begin{align*}
(D^2 (a\circ \phi^n) (x))_{i,j}&= \dd_i \dd_j (a\circ \phi^n) (x)\\
&= \sum _{k,\ell} ((\dd_k \dd_\ell (a\circ \phi^{n-1}))\circ \phi^1(x) )  \dd_j\phi^1_\ell(x) \dd_i \phi^1_k(x) \\
& \quad + \sum_k (\dd_k (a \circ \phi^{n-1}) \circ \phi^1(x) )\dd_j \dd_i \phi^1_k(x)
\end{align*}
so that 
\begin{align*}
    &D^2 (a\circ \phi^n )(x) \\
    &\quad =( D^2 (a\circ \phi^{n-1}) \circ \phi^1(x) )  (D\phi^1(x) \otimes D\phi^1(x)) + (D(a \circ \phi^{n-1}) \circ \phi^1(x)) D^2 \phi^1(x)
\end{align*}

Therefore
\begin{align*}
        \underbrace{\mat{  D( a \circ \phi ^{n+1})(x)  \\ D^2 ( a \circ \phi ^{n+1})(x)}^t}_{:=  A_{n+1}(x)^t} =  \mat{D(a \circ \phi^n)\circ \phi^1 (x) \\ D^2  (a \circ \phi^n)\circ \phi^1 (x)} ^t\underbrace{\mat{ D \phi^1 (x) & D^2 \phi^1 (x) \\ 0 & D\phi^1 (x) \otimes D\phi^1 (x)}}_{:= M( x)}  ,
\end{align*}
so that
\begin{align}\begin{split} \label{eq:1024}
    A_n(x) &= A_{n-1}(\phi^1 (x)) M(x)= A_{n-2} (\phi^2(x)) M(\phi^1 (x)) M(x))\\
    &=  A_0 (\phi^n(x)) M(\phi^{n-1}(x)) M(\phi^{n-2}(x)) \cdots M(x). \end{split}
\end{align}
By the upper triangular structure of $M$, it is easy to check that
\begin{align}
     M(\phi^{n-1}(x)) M(\phi^{n-2}(x)) \cdots M(x)= \mat{M_1(x) & M_2(x) \\ 0 & M_3(x)  } \label{eq:1030}
\end{align}
where:
\begin{align*}
    M_1(x) &:=\prod_{j=1}^{n}( (D\phi^1 )\circ ( \phi^{n-j}( x)) ), &&    M_2(x) :=  \sum_{j=1}^n (M_2)_j(x),   \\
    M_3(x) &= \prod_{j=1}^{n}( ((D\phi^1 )\circ \phi^{n-j}(x) )\otimes ((D\phi^1 )\circ \phi^{n-j}(x) )  )
\end{align*}
with $(M_2)_j (x)$ defined as
\begin{align}
     \left(\prod_{i=1}^{j-1}( (D\phi^1 )\circ ( \phi^{n-i}( x)) )\right) ((D^2 \phi^1 )\circ \phi^{n-j}(x))\left(\prod_{i = j+1}^{n}( (D\phi^1 )\circ ( \phi^{n-i}( x)) )^{\otimes 2}\right) \label{eq:m2j}
\end{align}
where $a^{\otimes 2} : = a \otimes a$. 

We then have, by \eqref{eq:1024} and \eqref{eq:1030},
\begin{align}\label{eq:1035}
\begin{split}
    |D^2 (a \circ \phi^n)(x) |&= |(Da\circ \phi^n (x))M_2(x) + (D^2 a \circ \phi^n(x) )M_3(x)|\\
    &\le C  \|a \|_{C^1} |M_2(x)| +  C\|a \|_{C^2} | M_3(x) |\end{split}
\end{align}
where the absolute value of a tensor denotes the maximum of the absolute value of all entries.

But now observe that 
\begin{align*}
    \prod_{i=1}^{j-1}( (D\phi^1 )\circ ( \phi^{n-i}( x)) &= D(\phi^{j-1}) \circ \phi^{n-j+1}(x),\\
    \prod_{i = j+1}^{n}( (D\phi^1 )\circ ( \phi^{n-i}( x)) )^{\otimes 2} &=     (D(\phi^{n-j})(x))^{\otimes 2}
\end{align*}
so that
\begin{align}
| M_3(x) | &= | D(\phi^n )(x) |^2 \le C e^{2 \lambda n} \label{eq:1044}
\end{align}
by the inductive hypothesis and setting $a$ to be coordinate functions. And we also have, by \eqref{eq:m2j}, that
\begin{align*}
| (M_2 )_j (x)| &\le C e^{\lambda(j-1)} e^{2\lambda(n-j)}\| D^2 \phi^1 \|_{L^\infty} =C' e^{2\lambda n -\lambda j-\lambda}
\end{align*}
where we use that $\phi^1$ is a smooth function on a compact set so that $\|D^2 \phi^1 \|_{L^\infty}<\infty$. But we stress that this norm only comes up \textit{once} in each $(M_2)_j$ term so that the constant at the end is \textit{independent} of $n$. We therefore have that:
\begin{align}\label{eq:1048}
|M_2(x) | \le \sum_{j=1}^n |(M_2)_j(x)| \le  C  e ^{2\lambda n - \lambda} \sum_{j=1}^n e^{-\lambda j }\le C e^{2\lambda n }(1- e^{-\lambda})^{-1}
\end{align}
We can absorb the $(1- e^{-\lambda})^{-1}$ term into the constant $C$ to get, using \eqref{eq:1035}, \eqref{eq:1044}, and \eqref{eq:1048}, that
\begin{align*}
     |D^2 (a \circ \phi^n)(x) | \le C \| a \|_{C^2} e^{2 \lambda n}.
\end{align*}
Now if the flow is continuous in time and $t\in \R _{\ge 0}$, we let $n\in \N$ be such that $n \le t \le n+1 $ and observe that:
\begin{align*}
    |D^2 (a\circ \phi^t) (x)| = |D^2 (a\circ \phi^{t-n} \circ \phi^n )(x) |\le C \| a \circ \phi^{t-n} \|_{C^2} e^{2\lambda n}.
\end{align*}
Note that $s:= t-n \in [0,1]$ which is a compact set, so that:
\begin{align}
    \| a \circ \phi^{t-n} \|_{C^2} \le \|a \|_{C^2} \max_{s\in [0,1]}\| \phi^s \|_{C^2}.
\end{align}
We can then absorb $\max_{s\in [0,1]}\| \phi^s \|_{C^2}$ into the constant $C$ to get the result for continuous time.

Now we present the general inductive step, but the reader should keep in mind that morally the idea is the same as above, but with more notation. Suppose we have for all $k'\le k-1$ and $|\alpha | = k'$:
\begin{align*}
|    \dd^\alpha (a \circ \phi^t ) (x)| \le C \| a \|_{C^{k'}} e^{k' \lambda t}.
\end{align*}
We then fix $n \in \mathbb N$, and aim to bound $D^{k} (a\circ \phi^n)(x)$. Define $A_n(x)$ as the $1\times k$ vector of tensors whose $i$th entry is $D^i (a \circ \phi^n)(x)$. By writing
\begin{align}
    D^i(a\circ \phi^n) = D^{i-1}( (D  (a \circ \phi^{n-1} ) \circ \phi^1) D\phi^1 )(x), \label{eq:1088}
\end{align}
and repeatedly applying the product rule, we can find a $k \times k$ matrix of tensors, $M(x)$ such that $A_n(x) = A_{n-1}(\phi^1(x)) M(x)$. Note that $M(x)$ will be upper triangular. Indeed, for any $i\le  k$, $D^i(a \circ \phi^n)$ is a linear combination of $D^j (a\circ \phi^{n-1}) \circ \phi^1$, for $j \le i$. Also, for each $i,j = 1, \dots, k$, the $i,j$ entry of $M(x)$ must be a $(i,j)$ tensor because it will be paired with a $(0,i)$ tensor to get a $(0,j)$ tensor. Each component of $(M(x))_{i,j}$ will contain $j$ derivatives of the components of $\phi^1(x)$ with a coefficient depending on $k$. Lastly note that for each $i = 1,\dots, k$, $(M(x))_{i,i} = (D\phi^1 (x))^{\otimes i}$ as this is the term in \eqref{eq:1088} when all derivatives fall on $D(a\circ \phi^{n-1}\circ \phi^1 )$.

We then recursively use the relation of $A_n$ to $A_{n-1}$ to get that:
\begin{align*}
    A_n(x) = A_0 (\phi^n(x))\left( \prod_{j=1}^n M(\phi^{n-j}(x))\right).
\end{align*}
For $i = 1,\dots ,k$, let $M_i$ be the $(i,k)$ entry of $\prod_{j=1}^n M(\phi^{n-j})$, so that:
\begin{align}
    D^k (a\circ \phi^n)(x) = \sum _{j=1}^k (D^j(a)\circ \phi^n(x) ) M_j(x). \label{eq:1077}
\end{align}
It now suffices to provide appropriate estimates of $M_j(x)$. For $i = 1,\dots ,n$ let $B^i : = M(\phi^{n-i})(x)$. Then for each fixed $j = 1, \dots, k$:
\begin{align}
    M_j =  \sum_{\ell_1,\ell_2,\dots, \ell_{n-1}} B^1_{j,\ell_1} B^2_{\ell_1,\ell_2} \cdots B^n_{\ell_{n-1},k}. \label{eq:Mj1045}
\end{align}
Because $B^i$ are upper triangular, this sum is nonzero only if $\ell_m \le \ell_{m+1}$ for each $m = 1,\dots, n-1$, and $j \le \ell_1$. Each term in the sum can be expressed as a path on a grid of $k\times (n+1)$ nodes. Labeling the bottom left node $(1,1)$, every path must begin at $(j,1)$ and end at $(k,n+1)$. The path can only go to the right or up. That is $(m,n)$ can only go to $(m',n+1)$ for $m' \ge  m$ (provided that the path terminates at $(k,n+1)$). The edges of the path give the index of $B^i$ in the sum. That is, if the nodes $(m,i),(n,i+1)$ are on the path, then we include the term $B^i_{m,n}$ in the sum.

For example, if $k = 4$ and $n= 5$, then the following path represents the term $B^1_{1,2}B^2_{2,2}B_{2,2}^3B^4_{2,4}B^5_{4,4}$.
\begin{center}
    \includegraphics[scale = .4]{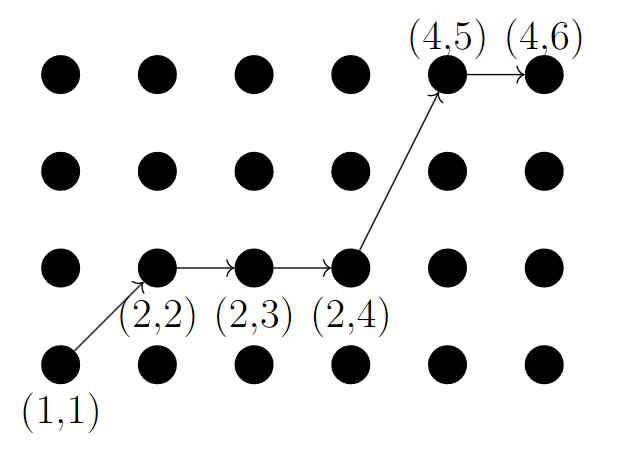}
\end{center}
Each path can be represented as a vector whose $i$th component is $m$ where $(m,i)$ is on the path, e.g. the example path is written as $(1,2,2,2,4,4)$. 

Each path will consist of flat sections (where the $i$ and $i-1$ components of the vector representation are the same) and jumps (where the component of the vector representation is not flat). In the above example, the jumps are at indices $2$ and $5$. The corresponding entry of $B^i$ for the jumps will be individually norm-bounded. Recall that $B^i_{j,\ell}$ will be a $(j,\ell)$ tensor where all terms are derivatives of $\phi^1$. As long as $(j,\ell) \neq (1,k)$, at most $k-1$ derivatives are taken of $\phi^1$. Note that $B_{1,k}$ can only appear once in any term in the sum \eqref{eq:Mj1045}, and this is bounded by the $C^k$ norm of $\phi^1$. We therefore can apply the inductive hypothesis to get that $|B^i_{j,\ell}(x)| \le C e^{\lambda \ell}$. Next, we bound the flat sequences. Suppose such a sequence is $\prod_{i=j_1}^{j_2}B_{\ell,\ell}^i$. Recall that $B_{\ell,\ell}^i= (D\phi^1)^{\otimes \ell}\circ (\phi^{n-i}(x))$, so that
\begin{align}
   \left| \prod_{i=j_1}^{j_2}B_{\ell,\ell}^i (x)\right|= | (D(\phi^{j_2-j_1+1})\circ \phi^{n-j_2+1} (x))^{\otimes \ell}|\le C e^{\ell(j_2-j_1+1)}. \label{eq:1092}
\end{align}
This follows by the inductive hypothesis by setting $a$ to be coordinate functions. Now suppose we have a term of the sum \eqref{eq:Mj1045} whose path has a vector representation $(a_1,a_2,\dots ,a_{n+1}) $ with jumps at indices $\ell_1,\ell_2,\dots, \ell_s$. We assume that there is a jump at index $2$ and the path ends in a flat section (as in the example), but all cases can be treated similarly. We can then bound this term of the sum (using \eqref{eq:1092}) by:
\begin{align*}
   & | B_{a_{\ell_1 - 1},a_{\ell_1}}^{\ell_1}| \cdots | B_{a_{\ell_s - 1},a_{\ell_s}}^{\ell_s}| \left | \prod_{i =\ell_1+1}^{\ell_2-1}B^i_{a_{\ell_1},a_{\ell_1}} \right | \cdots \left| \prod_{i = \ell_{s}+1}^{n+1}B^i_{a_{\ell_{s}},a_{\ell_{s}}} \right |\\
    &\le Ce^{\lambda(a_{\ell_1} + \cdots + a_{\ell_s})} e^{\lambda a_{\ell_1}  (\ell_2 - \ell_1 -1)} \cdots e^{\lambda a_{\ell_{s}}(n+1 - \ell_{s}-1)}\\
    &= C e^{\lambda (a_{\ell_1}  (\ell_2 - \ell_1 ) + \cdots + a_{\ell_{s-1}}(\ell_s - \ell_{s-1})  + a_{\ell_s} (n+1-\ell_s))}\\
    &= C e^{\lambda (a_2 + \cdots + a_{n+1})}.
\end{align*}
The last equality follows by noting that $\ell_{i+1} - \ell_i$ is the length of each flat section and is where the sequence has values $a_{\ell_i}$. Therefore:
\begin{align}
    |M_j(x) | \le C \left| \sum_{(a_n) \in \mathcal{A}_{j,k}}  e^{\lambda (a_2 + \dots +a_{n+1})}  \right| \label{eq:1103}
\end{align}
where $\mathcal{A}_{j,k}$ is the set of all sequences $(a_n)$ of length $n+1$ such that $a_i \in \mathbb Z_{ >0}$, $a_1= j$, $a_{n+1}= k$, and $a_i \le a_{i+1}$. Let $V$ be the $k\times k$ upper triangular matrix such that $V_{\ell,m} =  e^{m \lambda} $ for $\ell\le m$ and zero otherwise. Then by the exact same argument for computing products of upper triangular matrices, we see that
\begin{align}
   (V^n)_{j,k} =  \sum_{(a_n) \in \mathcal{A}_{j,k}}  e^{\lambda (a_2 + \dots +a_{n+1})} . \label{eq:1107}
\end{align}
Because the eigenvalues of $V$ are distinct, $V$ is diagonalizable. We can then write the entries of $V^n$ as linear combinations (with coefficients independent of $n$) of the eigenvalues of $V$ raised to the $n$th power. That is there exist $c_{j,\ell,k} \in \mathbb R$ independent of $n$ such that
\begin{align*}
    (V^n)_{j,k} = \sum _{\ell=1}^{k} e^{\lambda \ell n} c_{j,\ell,k}.
\end{align*}
But then we immediately see that there exists a $C = C(k)>0$ such that $|(V^n)_{j,k}| \le C e^{\lambda k n}$ for $j = 1, \dots ,k$. We therefore have that $|M_j(x) | \le C e^{ n \lambda k}$, so that \eqref{eq:1077} becomes:
\begin{align*}
    |D^k(a \circ \phi^n ) (x) | \le C \| a \|_{C^k} e^{n\lambda k}.
\end{align*}
Applying the exact same argument at the end of the case $k=2$ to go from discrete time to continuous time completes the proof.

\end{proof}

\printbibliography

\end{document}